\documentclass[11pt,letterpaper]{article}
\usepackage[utf8]{inputenc}
\usepackage[margin=1.0in]{geometry}
\usepackage{palatino}
\usepackage{macros}
\usepackage{algorithm}
\usepackage{algorithmic}
\usepackage{bm}
\usepackage{csquotes,fullpage}
\usepackage[T1]{fontenc}
\usepackage{subcaption}

\title{The Geometry of Lee-Yang Tensors}
\author{Jake Hofgard\thanks{UC Berkeley, Department of Mathematics, \url{hofgard@berkeley.edu}}}

\begin{document}
\maketitle 
\begin{abstract}
We provide a geometric description of Lee-Yang tensors, building upon the results from Bravyi, Gosset, Liu, and Wong in two recent papers \cite{wong2026, bravyi2026}. By applying results from complex Perron-Frobenius theory, we are able to generalize the ``squeezing lemma'' from \cite{bravyi2026} to bound the spectral ratio of non-normal operators with Lee-Yang radius greater than one. As a direct consequence, we provide a local-to-global principle for bounding spectral gaps of Suzuki-Fisher Hamiltonians, a family of $2$-local Hamiltonians, using only the Lee-Yang radius of each local operator to bound the gap from below. On the other hand, we develop a simple perturbative technique to show that the Gibbs states of $k$-local Hamiltonians, $k \geq 3$, cannot have Lee-Yang radius at least one at all inverse temperatures $\beta \geq 0$. We apply this same method to conclude that Suzuki-Fisher Hamiltonians are precisely those Hamiltonians whose Gibbs states are Lee-Yang at all inverse temperatures $\beta \geq 0$.
\end{abstract}
\tableofcontents

\section{Introduction}
Showing that a particular quantum Hamiltonian is gapped remains one of the most important guiding questions in mathematical physics \cite{Haldane1983, Tasaki2020, Tasaki2025}. Indeed, many physically interesting properties of a quantum many-body system can be deduced from its spectral gap, including (but not limited to) exponential decay of spatial correlations in the ground state of the system \cite{Hastings2006} and topologically protected phases of matter \cite{Bravyi2010}. The existence of a spectral gap also has important algorithmic implications, including access to efficient quantum algorithms for ground state preparation via adiabatic algorithms \cite{Kadowaki1998, Benseny2021}. For certain one-dimensional examples, the existence of a gap can also be translated into an efficient classical algorithm for ground state approximation \cite{Landau2015}. Nonetheless, even for Hamiltonians on the two-dimensional lattice, the problem of determining whether a family of Hamiltonians is gapped is undecidable \cite{Cubitt2022}.

Despite the general undecidability of proving that a family of Hamiltonians is gapped, the framework of Lee-Yang tensors has recently provided a foothold for showing that certain quantum Hamiltonians are gapped \cite{wong2026, bravyi2026, rayudu2026}. At a high-level, the framework of Lee-Yang tensors associates to every state and operator on the Hilbert space $(\C^2)^{\otimes n}$ a generating polynomial. By analyzing the roots of this generating polynomial, one can then deduce quantitative information about the spectral gap of a Hamiltonian. In particular, if the roots of the generating polynomial of a Hermitian operator $H$ all lie outside of a disk of radius $r > 1$, called its Lee-Yang radius, then \cite[Theorem 4]{wong2026} initially showed that the eigenvalue of $H$ with largest magnitude is non-degenerate.

The authors of \cite{wong2026} originated their study of Lee-Yang tensors to develop an efficient adiabatic quantum algorithm estimating the ground state of an antiferromagnetic Heisenberg model on a bipartite graph, sometimes called the quantum MaxCut (QMC) problem. This problem was later resolved by the same authors in \cite{bravyi2026} and simultaneously in \cite{rayudu2026}, with both papers using different techniques to establish the existence of a gap for the bipartite QMC Hamiltonian. Very recently, a similar guarantee was extended to an appropriate class of fermionic Hamiltonians in \cite{rayudu2026leeyangtheoremfermions}. In \cite{bravyi2026}, the authors directly quantify the spectral ratio of a positive definite operator $K$ in terms of its Lee-Yang radius. By instead considering the Lee-Yang radius of the \textit{Gibbs state} $\exp(-\beta H)$ of a Hamiltonian $H$, this spectral ratio result can be converted into a statement about the spectral gap of the Hamiltonian, assuming that it has uniformly positive $Z$-fields. This result applies to a broad class of Hamiltonians, called Suzuki-Fisher Hamiltonians, that were first introduced in \cite{Suzuki1971}. Importantly, the quantum MaxCut Hamiltonian on a bipartite graph belongs to this family after conjugating by an appropriate rotation on each site.

Our work was largely inspired by the results in \cite{bravyi2026}, many of which indicate the presence of a deeper geometric picture describing the fascinating properties of Lee-Yang tensors. We aim to characterize the geometry of Lee-Yang tensors by viewing operators with Lee-Yang radius greater than one as operators that preserve the cone of Lee-Yang tensors. This perspective allows us to access the rich theory of Perron-Frobenius operators from \cite{Rugh2010, dubois2010complexifiedconesspectralgaps} (Section~\ref{sec:hilbert-metrics}), generalizing the spectral ratio result from \cite[Theorem 2]{bravyi2026} to non-normal operators with Lee-Yang radius greater than one (Section~\ref{sec:contractions} and Section~\ref{sec:ly-connection}). Because the Suzuki-Trotter splitting of a Hamiltonian into its local terms yields a non-normal operator, this non-normal extension then gives a natural local-to-global principle for proving spectral gaps of Suzuki-Fisher Hamiltonians, a complementary perspective to the methods from \cite{bravyi2026} and \cite{rayudu2026} (Section~\ref{sec:local-to-global}). 

On the other hand, our geometric perspective on Lee-Yang tensors provides valuable insight into their limitations. For example, by analyzing the behavior of roots of the generating polynomial of the Gibbs state $\exp(-\beta H)$ as $\beta \to 0$, we show that $k$-local Hamiltonians with $k \geq 3$ cannot have Lee-Yang radius at least one for all inverse temperatures $\beta \geq 0$ (Section~\ref{sec:k-local}). In fact, this same method can be used to rather straightforwardly recover a surprising result that recently appeared in \cite{wei2026}: those Hamiltonians $H$ whose Gibbs states have Lee-Yang radius at least one for all inverse temperatures $\beta \geq 0$ are precisely the Suzuki-Fisher Hamiltonians introduced in \cite{Suzuki1971}. Consequently, although the techniques developed in \cite{bravyi2026, rayudu2026} and in Section~\ref{sec:ly-connection} of this paper are powerful, they cannot be extended beyond Suzuki-Fisher Hamiltonians without serious modification.

\subsection*{Related Work}
\paragraph{The Lee-Yang Theorem and Its Descendants.} In 1952, Lee and Yang initiated a research program that has since spurred decades of research in mathematical physics and more recently, quantum information. Their initial work  \cite{Lee1952} shows that the partition function of a ferromagnetic Ising model as a function of a uniform complex field $\mu \in \C$ only has zeros on the imaginary axis. As a consequence, the ferromagnetic Ising model cannot exhibit a phase transition for real, positive external fields. Nearly twenty years later, Asano showed in \cite{Asano1970} that the Lee-Yang theorem extends to Heisenberg ferromagnets and introduced Asano contractions, an elegant technique for proving such zero-freeness results. Contemporaneously, Ruelle simplified and generalized Asano's technique in the letter \cite{Ruelle1971}, lending the method its current name: Asano-Ruelle contractions. The work of Asano and Ruelle was closely followed by a related result of Suzuki and Fisher \cite{Suzuki1971} in which the authors proved a zero-freeness result for a much broader class of Heisenberg-type ferromagnets. 

In the years that followed, a number of authors applied similar techniques to even broader settings, including continuum lattice field theories \cite{Simon1973} and a general class of classical ferromagnetic spin systems \cite{Lieb1981}. It was also shown in \cite{Penrose1974, Dobrushin1987} that zero-freeness results in the style of the Lee-Yang theorem correspond to decay of correlations for classical spin systems.

\paragraph{Stable Polynomials and Stability-Preservers.} Later, the elegant work of Borcea and Brändén on the theory of stable polynomials recovered many of the aforementioned results as special cases of a complete characterization of operators that preserve stability; see \cite{BorceaBrandencircular2009, BorceaBrandenLeeYangI, BorceaBrandenLeeYang2009II} for the original results and \cite{Wagner2010} for a succinct survey. For instance, the technique of Asano-Ruelle contractions is a simple application of the machinery developed by Borcea and Brändén \cite[Lemma 8.8]{BorceaBrandenLeeYang2009II}. At the same time, Ruelle gave a complete characterization of Lee-Yang polynomials, defined therein as the class of complex polynomials $\Phi(z_1, \ldots, z_n)$ that do not vanish when $|z_1|, \ldots, |z_n| < 1$ or $|z_1|, \ldots, |z_n| > 1$ \cite{Ruelle2010}. Remarkably, Ruelle's characterization shows that the partition functions that satisfy a Lee-Yang theorem at all temperatures are \textit{precisely} the ferromagnetic Ising models originally considered in \cite{Lee1952}.

\paragraph{Lee-Yang Tensors.} As described above, a flurry of recent work has built upon the program pioneered by Lee and Yang, including (but not limited to) \cite{wong2026, bravyi2026, rayudu2026, rayudu2026leeyangtheoremfermions, wei2026}. Having already described the contributions of the first four of these works above, we only comment that the last paper, \cite{wei2026}, was released after we had already completed the majority of this work. We prove many of the same results as \cite{wei2026}, including the extension of \cite[Theorem 2]{bravyi2026} to non-normal operators via complex Perron-Frobenius theory. \cite{wei2026} also shows that Gibbs states of $k$-local Hamiltonians, $k \geq 3$, cannot be Lee-Yang at all temperatures using the language of stability preservers, a result that we obtain using a markedly different perturbative method. Finally, \cite{wei2026} also proves that any Hamiltonian whose Gibbs state is Lee-Yang at all temperatures must be a Suzuki-Fisher Hamiltonian. After learning about the result from \cite{wei2026}, we were able to recover it using the geometric methods that we developed for $k$-local Hamiltonians.

\section{Contributions}
In \cite{bravyi2026}, the authors establish that a broad class of Hamiltonians are gapped, paving the way for efficient adiabatic approximation of their ground state energies. The key technical contribution of this work is \cite[Lemma 1]{bravyi2026}, described as a ``squeezing lemma'' therein. In this note, we demonstrate that the squeezing lemma is a special case of a much broader category of results, originating from a decades-old effort to find a complex counterpart to the classical Perron-Frobenius theorem.

After providing relevant notation and background in Section~\ref{sec:notation}, Sections~\ref{sec:hilbert-metrics} and~\ref{sec:contractions} exposit some of the work from \cite{Rugh2010}, \cite{Dubois2009}, and \cite{dubois2010complexifiedconesspectralgaps}, drawing important connections to Hilbert's projective metric in hyperbolic geometry and Perron-Frobenius theory. Note that the results from these papers go far beyond what is needed to establish the squeezing lemma. In essence, the program of Rugh and Dubois centered around finding generalizations of the Perron-Frobenius theorem to complex matrices and integral operators. While many of their results establish that a certain operator or matrix has a spectral gap, some provide quantitative bounds on the gaps; see \cite[Theorem 1.1]{Dubois2009}, for instance. 

Section~\ref{sec:ly-connection} proves the squeezing lemma using the streamlined framework from \cite{Dubois2009} and discusses its interpretation in terms of the contractive properties of Lee-Yang tensors. The main result that we prove is a generalization of \cite[Theorem 2]{bravyi2026} to \textit{any} operator belonging to $\LY(r)$, the set of Lee-Yang tensors with radius at least $r > 1$, showing that all such operators admit spectral ratios of size at most $r^{-2}$. Notably, our result relies on the projective contraction principle of \cite{Dubois2009} to sidestep the power iteration argument from \cite{bravyi2026}, even when the operator in question is non-normal. 

\begin{theorem}
\label{th:lee-yang-gap}
If $r > 1$ and $K \in \LY(r)$ is an $n$-qubit operator with eigenvalues $|\lambda_1| \geq |\lambda_2| \geq \ldots \geq |\lambda_{2^n}|$, counted with multiplicity, then $\lambda_1$ has algebraic multiplicity one and any other eigenvalue $\lambda$ of $K$ satisfies
\begin{align*}
    \left|\frac{\lambda}{\lambda_1}\right| \leq r^{-2}.
\end{align*}
\end{theorem}

This allows us to formulate a local-to-global principle for Lee-Yang tensors in Section~\ref{sec:local-to-global}. In essence, to compute the spectral gap of a Lee-Yang operator that decomposes as a sum of local operators, it suffices to compute the Lee-Yang radius of each local operator. As long as each local operator has Lee-Yang radius satisfying $\log r(\beta) \geq c \beta + O(\beta^2)$ for some constant $c > 0$ as $\beta \to 0$, one can invoke Theorem~\ref{th:lee-yang-gap} to obtain quantitative control of the spectral gap of the global operator. This property holds, for example, for the EPR-like Hamiltonians described in \cite[Section 5]{wong2026}.

Finally, in Sections~\ref{sec:k-local} and~\ref{sec:ly-to-sf}, we show that Lee-Yangness is not a property that extends to $k$-local Hamiltonians. A perturbative analysis of the roots of the generating polynomial of $\exp(-\beta H)$ as $\beta \to 0$ reveals that, when $k \geq 3$, the Gibbs state $\exp(-\beta H)$ is \textit{not} Lee-Yang for sufficiently small $\beta$. This same method is used in Section~\ref{sec:ly-to-sf} to show that any Hamiltonian such that $\exp(-\beta H)$ is Lee-Yang for all $\beta \geq 0$ belongs to the Suzuki-Fisher class.

\subsection*{Acknowledgments}
The author would like to thank Cecilia Chen, João Basso, and Nikhil Srivastava for many helpful conversations, redirections, and guidance, especially during a time of immense change in mathematics and theoretical computer science. He also thanks Benjamin Wong and Alexander Zlokapa for helpful discussions when this work was in its infancy. This work was supported by the National Science Foundation Graduate Research Fellowship Program under Grant No. 2146752. This work was also done in part while the author was visiting the Simons Institute for the Theory of Computing.

\subsection*{AI Methodology}
We describe the ways in which generative AI tools contributed to this paper, following the recent guidelines in \cite{SimonsAITCS2026}.

\begin{center}
\renewcommand{\arraystretch}{1.3}
\begin{tabular}{@{}l@{\qquad}c@{}}
\textbf{Activity} & \textbf{AI used} \\
Asking the research question & No \\
Coming up with the approach & Yes \\
Proof development & Yes \\
Writing and exposition & No \\
Checking for bugs & Yes \\
Other supporting tasks & Yes \\
\end{tabular}
\end{center}

\begin{sloppypar}
\paragraph{Details.}
As mentioned above, this project originated as an effort to understand the geometric intuition behind \cite[Lemma 1]{bravyi2026}. The connection to complex Perron-Frobenius theory was discovered by ChatGPT-5.6 Sol after an extended back-and-forth. The author had also discussed this problem with ChatGPT-5.6 in previous chats (which may have served as context), although these previous conversations did not reveal the connection to the work of Birkhoff, Rugh, or Dubois that is described below. After this connection was uncovered, Lemma~\ref{lem:bounded-oscillation} and Lemma~\ref{lem:K-projective-diameter} were produced during another work session with ChatGPT-5.6, with the goal of proving a non-normal extension of \cite[Theorem 2]{bravyi2026}.

The local-to-global principle in Section~\ref{sec:local-to-global} is the author's own work, as is the perturbative technique via Rouché's theorem in Section~\ref{sec:k-local}. However, the general obstruction to Lee-Yangness at all temperatures for $k$-local Hamiltonians was developed in collaboration with ChatGPT-6 Astra, beginning with the simple example of a $k$-local projector Hamiltonian that is gapped but whose Gibbs state is not Lee-Yang at high enough temperature. Similarly, after learning that the Hamiltonians whose Gibbs states are Lee-Yang at all temperatures are precisely the Suzuki-Fisher Hamiltonians from \cite{wei2026}, the author sought to recover this result using the elementary perturbative technique from Section~\ref{sec:k-local}. The initial details of Section~\ref{sec:ly-to-sf} were developed using ChatGPT-6 Astra, including preliminary versions of Lemma~\ref{lem:ly-to-sf} and Lemma~\ref{lem:ly-to-positive-field}. The code for Figure~\ref{fig:root-shattering} was written by ChatGPT-6 Astra based on a description of the desired visualization.

All of the writing in this paper is the author's own, with the primary goal of providing a clear, intuitive explanation of the fascinating geometry underlying Lee-Yang tensors. Claude Opus 5.5 was used to proofread and correct lingering minor errors and ambiguities, especially in Section~\ref{sec:hilbert-metrics} and in the proof of Lemma~\ref{lem:projective-cauchy-sequence}. The author accepts responsibility for the presentation and correctness of all material in this paper.
\end{sloppypar}

\section{Preliminaries and Notation}
\label{sec:notation}
Unless otherwise specified, we work over the complex vector space $(\C^2)^{\otimes n}$, corresponding to a quantum spin system of $n$ qubits over an interaction graph $G = (V, E)$ with vertices enumerated as $V = \{1, \ldots, n\}$. In Section~\ref{sec:k-local}, when we describe the limitations of Lee-Yangness in the setting of $k$-local Hamiltonians, the interaction graph is replaced by an interaction hypergraph. We defer the necessary notation for $k$-local Hamiltonians to Section~\ref{sec:k-local}.

On $(\C^2)^{\otimes n}$, we denote by $\|\psi\| := \|\psi\|_2$ the standard Euclidean norm, and for operators $A \in \mathcal{B}((\C^2)^{\otimes n}) \cong (\C^2)^{\otimes 2n}$, we consider the induced operator norm
\begin{align*}
    \|A\| := \sup_{\|\psi\| = 1} \|A \psi\|.
\end{align*}
For certain technical results, we also consider the Hilbert-Schmidt norm of an operator $A$, defined by
\begin{align*}
    \|A\|_{\mathrm{HS}}^2 = \Tr(A^* A),
\end{align*}
where $A^*$ is the Hermitian adjoint of $A$.

We now recall the main definitions and results from \cite{wong2026} and \cite{bravyi2026} that we require below.
\begin{definition}
\label{def:generating-poly}
Given a state $\psi \in (\C^2)^{\otimes n}$ with entries indexed by $x \in \{0,1\}^n$, we define a generating polynomial $f_\psi : \C^n \to \C$ by 
\begin{align}
\label{eq:gen-poly}
    f_\psi(z) := \sum_{x \in \{0,1\}^n} \psi_x \prod_{a \in \supp(x)} z_a.
\end{align}
Here, $z = (z_1, \ldots, z_n)$ and $\supp(x) = \{i \in [n] : x_i = 1\}$.
\end{definition}

Note that, in the computational basis, we may write
\begin{align*}
    \ket{\psi} = \sum_{x \in \{0,1\}^n} \psi_x \ket{x}
\end{align*}
In turn, if we define
\begin{align*}
    \ket{v(z)} := \bigotimes_{i=1}^n (\ket{0} + z_i^* \ket{1}),
\end{align*}
then
\begin{align*}
    f_\psi(z) = \langle v(z) | \psi \rangle.
\end{align*}
For operators $A : (\C^2)^{\otimes n} \to (\C^2)^{\otimes n}$, we may either view $A$ as an element of $(\C^2)^{\otimes 2n}$ and apply Definition~\ref{def:generating-poly} directly \textit{or} note that
\begin{align*}
    F_A(z, w) = \langle v(z) | A | v(w^*) \rangle,
\end{align*}
where $w^* = (w_1^*, \ldots, w_n^*)$, in line with the physics convention that the second slot of the inner product is linear. Typically, we use lowercase $f$ to denote the generating polynomial of a state and uppercase $F$ to denote the generating polynomial of an operator.

We consider the zeros of generating polynomials of both Gibbs states and ground states of certain Hamiltonians, defined below.
\begin{definition}
\label{def:ly-tensor}
A state $\psi \in (\C^2)^{\otimes n}$ is a \textbf{Lee-Yang tensor} with radius $\mathbf{r} \in \R_+^n$ if the generating polynomial of $\psi$ has no zeros in the open centered multi-disk of radius $\mathbf{r}$, given by
\begin{align*}
    \D_{\mathbf{r}} := \D_{r_1} \times \ldots \times \D_{r_n}.
\end{align*}
Any state that satisfies this property is said to belong to $\LY_n(\mathbf{r})$.
For $r > 0$, we denote by $\LY(r)$ the set of all Lee-Yang tensors with radius at least $r$ for each variable, given by
\begin{align*}
    \LY(r) := \bigcup_{n \geq 0} \LY_n((r, \ldots, r)).
\end{align*}
\end{definition}

The following result is essentially a modern formulation of Asano--Ruelle contractions, from \cite{Asano1970, Ruelle1971}, in the language of Definition~\ref{def:ly-tensor}.

\begin{proposition}[Lemma 1 in \cite{wong2026}]
\label{prop:ly-contraction}
Suppose that $\psi \in \LY_n(\mathbf{r})$ for some $n \geq 2$. If $\phi \in (\C^2)^{\otimes (n - 2)}$ is the tensor obtained from $\psi$ by contracting a pair of indices $i < j$ such that $r_i r_j > 1$, then $\phi \in \LY_{n-2}(\mathbf{r'})$, where $\mathbf{r'}$ is the result of deleting entries $i$ and $j$ of $\mathbf{r}$. If $r_i r_j = 1$, then either $\phi$ is identically zero or $\phi \in \LY_{n-2}(\mathbf{r'})$.
\end{proposition}

A key consequence of Proposition~\ref{prop:ly-contraction} is the fact that if $A, B$ are operators such that $A, B \in \LY(r)$ for some $r > 1$, then $A \otimes B \in \LY(r)$, and in turn, $AB \in \LY(r)$. Indeed, it is easy to check that the generating polynomial of $A \otimes B$ factorizes as the generating polynomial of $A$ times the generating polynomial of $B$. Because $AB$ is the result of contracting $n$ pairs of indices in $A \otimes B$, it follows immediately from Proposition~\ref{prop:ly-contraction} that $AB \in \LY(r)$.

Finally, we introduce the family of Suzuki-Fisher Hamiltonians, a restricted class of anisotropic Heisenberg models, that feature heavily in Section~\ref{sec:ly-to-sf}. We borrow the modernized definition from \cite[Definition 1]{bravyi2026}, recalling that this class of Hamiltonians was initially introduced by the eponymous authors all the way back in 1971 \cite{Suzuki1971}.

\begin{definition}[Suzuki-Fisher Hamiltonians]
\label{def:suzuki-fisher}
An $n$-qubit Hermitian operator $H$ is a Suzuki-Fisher Hamiltonian if it can be written as
\begin{align*}
    H = -\sum_{\{i,j\} \in E} (J_{ij}^{zz} \sigma_{z}^i \sigma_z^j + J_{ij}^{xx}  \sigma_{x}^i \sigma_x^j + J_{ij}^{yy}  \sigma_{y}^i \sigma_y^j + J_{ij}^{xy}  \sigma_{x}^i \sigma_y^j + J_{ij}^{yx}  \sigma_{y}^i \sigma_x^j) - \sum_{i \in V} (\mu_i^x \sigma_x^i + \mu_i^y \sigma_y^i + \mu_i^z \sigma_z^i),
\end{align*} with $\mu_i^z \geq 0$ for all $i \in V$ and
\begin{align}
    J_{ij}^{zz} \geq \left\|\begin{pmatrix}
    J_{ij}^{xx} & J_{ij}^{xy} \\ J_{ij}^{yx} & J_{ij}^{yy}
    \end{pmatrix}\right\|
\end{align}
for all $\{i, j\} \in E$.
\end{definition}
Note that, due to the Hermitianity of $H$, all of the field and coupling coefficients are real in addition to the restrictions imposed upon them by Definition~\ref{def:suzuki-fisher}. In Section~\ref{sec:ly-to-sf}, we also permit the presence of a term of the form $c I$, where $c \in \R$, in Suzuki-Fisher Hamiltonians. Note that the addition of such a term does \textit{not} alter the Lee-Yang radius of the Gibbs state $\exp(-\beta H)$.

Although we don't invoke it directly, we also present \cite[Theorem 1]{bravyi2026} for the sake of completeness. Note that Theorem~\ref{th:lee-yang-gap} reduces to \cite[Theorem 2]{bravyi2026} in the case that $K \in \LY(r)$ is positive semidefinite. Our framework can thus recover the following result, although we do not pursue this connection further. A similar result is also presented in \cite[Theorem 11]{rayudu2026}.

\begin{theorem}[Theorem 1 in \cite{bravyi2026}]
\label{th:bravyi-gap}
Suppose that $H$ is a Suzuki-Fisher Hamiltonian with eigenvalues $\lambda_1 \leq \lambda_2 \leq \ldots \leq \lambda_{2^n}$. Then, if $\mu = \min_{i \in V} \mu_i^z$ is the smallest $Z$-field, $H$ has spectral gap at least $\lambda_2 - \lambda_1 \geq 2 \mu$.
\end{theorem}

\section{Hilbert's Projective Metric and Birkhoff's Theorem}
\label{sec:hilbert-metrics}
In 1957, Birkhoff introduced an elegant framework that greatly generalizes the classical Perron-Frobenius theorem \cite{Birkhoff1957}. This framework, although not directly related to Lee-Yang tensors, is worth understanding, as it forms the foundation of the work of Dubois and Rugh in the setting of complex cones. The overview \cite{lemmens2013} provides a comprehensive synopsis of Hilbert's projective metric and Birkhoff's application of it to Perron-Frobenius theory; we follow this survey below. The connection to Markov chain mixing in \cite{cohen2024} is another valuable resource, especially for the applications that we have in mind.

As a warmup, consider the positive quadrant in $\R^2$, denoted by
\begin{align*}
    \R^2_+ := \{(x,y) \in \R^2 : x > 0 \text{ and } y > 0\}.
\end{align*}
Points in $\R^2_+$ can be projected onto the line segment joining $(0, 1)$ and $(1, 0)$, via the map
\begin{align*}
    (f_1(x,y), f_2(x,y)) = \frac{1}{|x| + |y|} (x, y).
\end{align*}
In turn, we can map the resulting line segment onto the hyperbolic line by taking $u = \log(f_2/ f_1)$. We can then define a metric on the line segment, taking
\begin{align*}
    \rho(f, g) = \left|\log\left(\frac{f_2 g_1}{f_1 g_2}\right)\right|,
\end{align*}
which is induced by the usual metric $|u - v|$ on the hyperbolic line. The line segment above is precisely the projective cone $\R^2_+ / \! \sim$, where $(x_1, y_1) \sim (x_2, y_2)$ if and only if they project to the same point on the line segment. The metric defined above is typically called Hilbert's (cross-ratio) metric.

As outlined in \cite{Bushell1973, lemmens2013}, Birkhoff generalized Hilbert's metric to cones over real vector spaces. Birkhoff's version of Hilbert's metric provides a natural foundation for generalizing the results in \cite{bravyi2026}. In the following section, we will instead consider complex vector spaces, but the background in this section is helpful for interpreting some of the definitions that arise in the complex setting. Following the conventions in \cite[Section 2]{lemmens2013}, let $V$ be a real vector space. Recall that $C \subseteq V$ is a (pointed) cone if $C$ is convex, $\lambda C \subseteq C$ for all $\lambda \geq 0$, and $C \cap (-C) = \{0\}$. A cone $C \subseteq V$ then induces a partial ordering $\leq_C$ on $V$ by $x \leq_C y$ if and only if $y - x \in C$.

\begin{definition}
\label{def:cone-domination}
If $x \in V$ and $y \in C$, we say that $y$ dominates $x$ if there exist $\alpha, \beta \in \R$ such that $\alpha y \leq_C x \leq_C \beta y$. In this case, take
\begin{align*}
    M_C(x,y) := \inf\{\beta \in \R : \beta y - x \in C\}, \quad
    m_C(x,y) := \sup\{\alpha \in \R : x - \alpha y \in C\}.
\end{align*}
\end{definition}

If $y \in C \setminus \{0\}$ and $y$ dominates $x$, then $M_C(x,y) \geq m_C(x,y)$. In fact, we can construct an equivalence relation on $C$ by $x \sim_C y$ if $y$ dominates $x$ and $x$ dominates $y$. One can check that if $x, y \in C \setminus \{0\}$, then $x \sim_C y$ if and only if there exist $0 < \alpha \leq \beta$ such that $\alpha y \leq_C x \leq_C \beta y$, and in this case,
\begin{align*}
    m_C(x, y) = \sup\{\alpha > 0 : y \leq_C \alpha^{-1} x\} = M_C(y, x)^{-1}.
\end{align*}
We now define Birkhoff's version of Hilbert's metric on $C$ as
\begin{align*}
    d(x, y) := \log\left(\frac{M_C(x, y)}{m_C(x, y)}\right).
\end{align*}
See \cite[Lemma 2.1]{lemmens2013} to verify that this is indeed a metric.

To motivate the complex setting of Section~\ref{sec:contractions}, it is instructive to rephrase $M_C$ and $m_C$ in terms of what we will call an \textit{exceptional set}. For linearly independent $x, y \in C$ with $x \sim_C y$, let
\begin{align*}
    E_C(x, y) := \{\alpha \in \R : \alpha x - y \notin C \cup (-C)\}.
\end{align*}
Because $\alpha x - y \in C$ if and only if $y \leq_C \alpha x$, and $\alpha x - y \in -C$ if and only if $\alpha x \leq_C y$, it follows that $E_C(x, y)$ is an interval with endpoints $M_C(x,y)^{-1}$ and $m_C(x,y)^{-1}$. Consequently,
\begin{align*}
    d(x, y) = \log\left(\frac{\sup E_C(x,y)}{\inf E_C(x,y)}\right).
\end{align*}
For example, let $C = \R^2_+$, with $x = (1,1)$ and $y = (1,2)$. Observe that $\alpha x - y = (\alpha - 1, \alpha - 2)$ lies in $C$ exactly when $\alpha \geq 2$ and in $-C$ exactly when $\alpha \leq 1$, so $E_C(x, y) = (1, 2)$. Hence $M_C(x,y) = 1$, $m_C(x,y) = 1/2$, and $d(x,y) = \log 2$. In agreement with Hilbert's metric, we have that $|\log(x_2 y_1 / x_1 y_2)| = \log 2$. In Section~\ref{sec:contractions}, complex cones are instead invariant under multiplication by $-1$ (and all other nonzero complex scalars), so the condition $\alpha x - y \notin C \cup (-C)$ in the definition of the exceptional set simplifies to $zx - y \notin C$.

Setting aside the notion of an exceptional set momentarily, we return to \cite{Birkhoff1957}. Therein, Birkhoff sought to understand linear mappings between cones. If $V, W$ are real vector spaces with $C \subseteq V$ and $K \subseteq W$ cones, then for a linear map $L : V \to W$ with $L(C) \subseteq K$, the contraction ratio of $L$ is given by
\begin{align*}
    \kappa(L) := \inf\{c \geq 0: d(L x, L y) \leq c d(x, y) \text{ for all } x \sim_C y\}.
\end{align*}
In turn, the projective diameter of a linear map $L$ is given by
\begin{align*}
    \Delta(L) := \sup\{d(Lx, L y) : x, y \in C, \; Lx \sim_K L y\}.
\end{align*}
Birkhoff's influential result from \cite{Birkhoff1957} then quantifies the contraction ratio of a linear map in terms of its projective diameter. We quote a concise statement of the theorem from \cite[Theorem 2.9]{lemmens2013}:

\begin{theorem}[Birkhoff's Theorem]
\label{th:birkhoff}
If $C$ is a cone in a real vector space $V$ and $K$ is a cone in a real vector space $W$, with $L : V \to W$ a linear mapping such that $L(C) \subseteq K$, then
\begin{align*}
    \kappa(L) = \tanh\left(\frac{1}{4} \Delta(L)\right).
\end{align*}
\end{theorem}

Stated informally, a linear map $L$ between cones is a contraction in the Hilbert metric, and the contraction can be quantified precisely in terms of the projective diameter of $L$. To foreshadow the relevance of this result, recall that we are attempting to quantify spectral gaps of operators with Lee-Yang radius at least one. In \cite{bravyi2026}, the key result towards obtaining a gap was the so-called ``squeezing lemma,'' which we aim to recast in terms of a result such as Theorem~\ref{th:birkhoff}. As we will see in Section~\ref{sec:contractions}, the set of Lee-Yang tensors forms a \textit{complex} cone, mapped to itself by operators with Lee-Yang radius at least one. Hence, we may anticipate that the operators with Lee-Yang radius at least one are contractions on the cone of Lee-Yang tensors via an appropriate analogue of Theorem~\ref{th:birkhoff}. Through this lens, we will be able to reprove and generalize \cite[Lemma 1]{bravyi2026}.

\section{Complex Cones and Projective Contractions}
\label{sec:contractions}
In extending the results in the preceding section to complex cones, we follow the setup in \cite[Theorem 2.3]{Dubois2009}, as it is most closely connected to the squeezing lemma from \cite{bravyi2026} and we found Dubois' presentation to be a bit easier to follow. However, essentially the same result can be extracted from \cite[Lemma 2.4]{Rugh2010} with little difficulty. We work on complex cones, defined as subsets of a complex Banach space $V$ if we wish to be as general as possible. The discussion below begins in this somewhat abstract setting, but we will quickly focus on the application of the theory from \cite{Dubois2009} to the complex cone of Lee-Yang tensors.
\begin{definition}[Definition 1 in \cite{Dubois2009}]
\label{def:complex-cone}
A non-empty subset $C \subset V$ is a complex cone if, for all $z \in \C^* := \C \setminus \{0\}$, we have that $z C \subseteq C$. Such a cone is called \textit{proper} if $\overline{C}$ contains no complex planes.
\end{definition}
For proper complex cones, we may define an exceptional set (in the language of \cite{dubois2010complexifiedconesspectralgaps}) for points $x, y \in C \setminus \{0\}$ as follows:
\begin{align}
\label{eq:exceptional-set}
E_C(x, y) := \{z \in \C : zx - y \notin C\}.
\end{align}
Note that, for proper cones, we always have that $E_C(x, y) \neq \emptyset$ for $x, y$ linearly independent, as $C$ would otherwise contain a complex plane. Moreover, if $0 \in E_C(x, y)$, then $y \notin C$, a contradiction. Thus, $E_C(x, y)$ is a non-empty subset of $\C^*$ for all $x, y \in C \setminus \{0\}$ that are linearly independent.

Using the notion of the exceptional set, we can define a projective pre-metric on the cone $C$.
\begin{align}
\label{eq:projective-metric}
\delta_C(x, y) = \begin{cases}
0 & y = zx \text{ for some } z \in \C, \\
\log\left(\frac{\sup|E_C(x, y)|}{\inf|E_C(x,y)|}\right) & \text{otherwise}.
\end{cases}
\end{align}
For arbitrary complex cones, $\delta_C$ is a pre-metric in the sense that it does \textit{not} necessarily satisfy the triangle inequality. However, $\delta_C$ is always symmetric and satisfies $\delta_C(\lambda x, \lambda y) = \delta_C(x, y)$ for all $\lambda \in \C^*$. To see why this last property holds, observe that
\begin{align*}
    E_C(x, \lambda y) = \{z \in \C : zx - \lambda y \notin C\}.
\end{align*}
Now, if $z \in E(x, \lambda y)$ is such that $z \notin E(x / \lambda, y)$, then $z x / \lambda - y \in C$. In turn, because $C$ is a complex cone, we have that $zx - \lambda y \in \lambda C \subseteq C$, a contradiction of the fact that $z \in E_C(x, \lambda y)$. Hence, $E_C(x, \lambda y) \subseteq E_C(x / \lambda, y)$. Conversely, if $z \in E_C(x / \lambda, y)$, then the same reasoning shows that $z \in E_C(x, \lambda y)$ so that the two sets are in fact equal. Hence, $\delta_C(x, \lambda y) = \delta_C(x / \lambda, y)$ for all $\lambda \in \C^*$. In particular,
\begin{align*}
    \delta_C(x, y) = \delta_C(x, \lambda(y / \lambda)) = \delta_C(x / \lambda, y / \lambda).
\end{align*}
As this holds for all $\lambda \in \C^*$, the claim follows. In fact, as we will show below, $\delta_C(\lambda_1 x, \lambda_2 y) = \delta_C(x, y)$ for all $\lambda_1, \lambda_2 \in \C^*$. This property explains why $\delta_C$ is referred to as a projective (pre)metric on the complex cone $C$: it is scale-invariant. In order to turn $\delta_C$ into a proper metric, we require an additional property of the cone $C$.
\begin{definition}
\label{def:linearly-convex-cone}
A complex cone $C \subset V$ is \textit{linearly convex} if for every point $v \in V \setminus C$, there exists a hyperplane passing through $v$ that does not intersect $C$.
\end{definition}
If $V'$ is the dual space to $V$ and $\langle \cdot, \cdot \rangle : V' \times V \to \C$ is the canonical pairing, then we can define the \textit{dual complement} of $C$ to be the set
\begin{align*}
    C' := \{f \in V' : \langle f, x \rangle \neq 0 \text{ for all } x \in C\}.
\end{align*}
For linearly convex cones, the following property holds.
\begin{lemma}
\label{lem:dual-complement-char}
A point $x \in V$ belongs to a linearly convex cone $C$ if and only if for all $f \in C'$, $\langle f, x \rangle \neq 0$.
\end{lemma}
\begin{proof}
If $x \in C$, it follows directly from the definition of $C'$ that $\langle f, x \rangle \neq 0$ for all $f \in C'$. Conversely, suppose that $x \notin C$. Then, because $C$ is a linearly convex cone, there exists a complex hyperplane passing through $x$ that does not intersect $C$. In particular, there exists $f \in V'$ such that $\langle f, x \rangle = 0$ but $\langle f, y \rangle \neq 0$ for all $y \in C$. In other words, there exists $f \in C'$ with $\langle f, x \rangle = 0$. Taking the contrapositive of what was just shown reveals that if $\langle f, x \rangle \neq 0$ for all $f \in C'$, then $x \in C$.
\end{proof}
This property gives an alternative characterization of $\delta_C$ when $C$ is linearly convex, upgrading $\delta_C$ from a pre-metric to a metric.
\begin{lemma}[Lemma 2.1 in \cite{Dubois2009}]
\label{lem:linearly-convex-metric}
If $C$ is a linearly convex proper complex cone, then $\delta_C$ is a projective metric on $C$, and it can be variationally characterized as
\begin{align*}
    \delta_C(x, y) = \sup_{f,g \in C'} \log \left|\frac{\langle f, y \rangle \langle g, x \rangle}{\langle f, x \rangle \langle g, y \rangle}\right|.
\end{align*}
\end{lemma}
\begin{proof}
If $x, y \in C$, then by Lemma~\ref{lem:dual-complement-char}, we know that $zx - y \notin C$ if and only if $\langle f, zx - y \rangle = 0$ for some $f \in C'$. In turn,
\begin{align*}
    E_C(x, y) = \left\{\frac{\langle f, y \rangle}{\langle f, x \rangle} : f \in C'\right\}.
\end{align*}
Recalling the definition of the exceptional set from Equation~\eqref{eq:exceptional-set} gives the claimed formula for $\delta_C$ when $x$ and $y$ are not collinear. When $x, y \in C$ \textit{are} collinear so that $y = cx$ for some $c \in \C^*$, we have that
\begin{align*}
    \frac{\langle f, y \rangle}{\langle f, x \rangle} = c
\end{align*}
for all $f \in C'$, so the supremum of the above quantity over all $f, g \in C'$ is zero.

We can now recover the triangle inequality because if $x, y , z \in C \setminus \{0\}$, then for all $f, g \in C'$,
\begin{align*}
    \log \left|\frac{\langle f, y \rangle \langle g, x \rangle}{\langle f, x \rangle \langle g, y \rangle}\right| &= \log\left|\frac{\langle f, z \rangle \langle g, x \rangle}{\langle f, x \rangle \langle g, z \rangle} \cdot \frac{\langle f, y \rangle \langle g, z \rangle}{\langle f, z \rangle \langle g, y \rangle}\right| \\
    &= \log\left|\frac{\langle f, z \rangle \langle g, x \rangle}{\langle f, x \rangle \langle g, z \rangle} \right| + \log \left|\frac{\langle f, y \rangle \langle g, z \rangle}{\langle f, z \rangle \langle g, y \rangle}\right|.
\end{align*}
Taking the supremum over all $f, g \in C'$ then yields $\delta_C(x, y) \leq \delta_C(x, z) + \delta_C(y, z) = \delta_C(x, z) + \delta_C(z, y)$. Note that the fact that $\delta_C(x, y) = \delta_C(y, x)$ is apparent from the above characterization, although symmetry still holds even when $C$ is not linearly convex.
\end{proof}
Writing $x \sim y$ if and only if $\C^* x = \C^* y$, we can define the projective metric space $(C / {\sim}, \delta_C)$ for proper, linearly convex complex cones $C$. Under an additional assumption on the cone $C$ (see \cite[Definition 4]{Dubois2009}), this is a complete metric space. If $V$ is finite-dimensional, the additional regularity requirement on $C$ (that it has bounded sectional aperture) is always satisfied as shown in \cite[Lemma 8.2]{Rugh2010}, and we may freely invoke the following result.
\begin{lemma}[Lemma 2.6 in \cite{Dubois2009}]
\label{lem:projective-complete-metric-space}
If $C$ is a proper, linearly convex complex cone and $V$ is a finite-dimensional vector space over $\C$, then the projective space $(C / {\sim}, \delta_C)$ is a complete metric space.
\end{lemma}

Although we will invoke this result once below in the setting of Lee-Yang tensors, we first turn to the content of \cite[Theorem 2.3]{Dubois2009}. In a nutshell, Dubois' result states that any linear map $T$ between proper complex cones is a contraction in projective space. Moreover, the scale of the contraction can be bounded explicitly in terms of the projective diameter of $T$.
\begin{theorem}[Projective Contraction Principle, Theorem 2.3 in \cite{Dubois2009}]
\label{th:projective-contraction-principle}
Suppose that $V_1, V_2$ are complex Banach spaces and $C_1 \subset V_1$ and $C_2 \subset V_2$ are proper complex cones. If $T : V_1 \to V_2$ is a linear map with $T(C_1 \setminus \{0\}) \subseteq C_2 \setminus \{0\}$ and
\begin{align*}
    \Delta_T := \sup_{x, y \in C_1 \setminus \{0\}} \delta_{C_2}(Tx, Ty)
\end{align*}
is finite, then
\begin{align}
\label{eq:contraction}
    \delta_{C_2}(T x, T y) \leq \tanh\left(\frac{\Delta_T}{4}\right) \delta_{C_1}(x, y)
\end{align}
for all $x, y \in C_1$.
\end{theorem}
The key step in the proof of Theorem~\ref{th:projective-contraction-principle} is the rather straightforward observation that if $Tx$ and $Ty$ are linearly independent (if not, the contraction principle holds trivially), then
\begin{align*}
    z x - y \in C_1 \implies z T x - T y \in C_2
\end{align*}
because $T(C_1 \setminus \{0\}) \subseteq C_2 \setminus \{0\}$. Hence, $E_{C_2}(Tx, Ty) \subseteq E_{C_1}(x, y)$. In other words, $T$ contracts exceptional sets. In fact, to prove the squeezing lemma, this is the \textit{only} fact that we will need!

\section{Applications to Lee-Yang Tensors}
\label{sec:ly-connection}
We now turn our attention to Lee-Yang tensors. In particular, consider the set $\LY(r)$ of Lee-Yang tensors with radius $r$, viewed as a subset of the complex vector space $(\C^2)^{\otimes n}$. Throughout this section, we fix $V := (\C^2)^{\otimes n}$ and abuse notation by writing $\LY(r)$ for $\LY(r) \cap V = \LY_n((r, \ldots, r))$. For an operator $K$ on $V$, we write $K \in \LY(r)$ to mean that $K$, viewed as an element of $(\C^2)^{\otimes 2n}$, belongs to $\LY_{2n}((r, \ldots, r))$.
\begin{proposition}
\label{prop:ly-tensors-cone}
The set $\LY(r) \subset V$ is a proper, linearly convex complex cone.
\end{proposition}
\begin{proof}
If $\psi \in \LY(r)$, then $\lambda \psi \in \LY(r)$ for all $\lambda \in \C^*$ because
\begin{align*}
    f_{\lambda \psi}(z) = \lambda f_{\psi}(z).
\end{align*}
Hence, if $f_\psi$ is non-vanishing on $\D_r^n$, then so is $f_{\lambda \psi}$ for any $\lambda \in \C^*$.

To see that $\LY(r)$ is proper, it suffices to show that it contains no complex planes because $\LY(r) \cup \{0\}$ is closed. If $\psi, \phi \in \LY(r)$ are linearly independent, then for any $z \in \D_r^n$
\begin{align*}
    \rho := f_\psi(z) \phi - f_\phi(z) \psi
\end{align*}
is non-zero (by linear independence), but $f_{\rho}(z) = 0$ so that $\rho \notin \LY(r)$. Thus, the complex plane spanned by $\psi, \phi$ cannot lie in $\LY(r)$.

Finally, to see that $\LY(r)$ is linearly convex, take any point $\psi \notin \LY(r)$. For such $\psi$, there exists $z_\psi \in \D_r^n$ such that $f_\psi(z_\psi) = 0$. In turn, the hyperplane $\{\phi \in V : f_{\phi}(z_{\psi}) = 0\}$ passes through $\psi$ and avoids $\LY(r)$.
\end{proof}

For the purposes of understanding the squeezing lemma, it now suffices to focus on the following slightly more general setting in which $K : V \to V$ is a linear map belonging to $\LY(r)$. In particular, it follows that $K(\LY(s)) \subseteq \LY(r)$ for all $s > 1/r$ so that $K$ satisfies the hypotheses of Theorem~\ref{th:projective-contraction-principle}. In this setting, a slightly asymmetric generalization of Equation~\eqref{eq:exceptional-set} provides valuable information about the eigenvectors of $K$.

\begin{lemma}
\label{lem:ly-exceptional-set}
If $K$ is as above and $\phi, \psi \in V \setminus \{0\}$ are linearly independent with $\psi \in \LY(r)$, then
\begin{align*}
    E_{\LY(s)}(\psi, \phi) := \{\lambda \in \C : \lambda \psi - \phi \notin \LY(s)\} = \left\{\frac{f_\phi(z)}{f_\psi(z)} : z \in \D_s^n \right\}
\end{align*}
for all $s \leq r$, and
\begin{align*}
    E_{\LY(r)}(K \psi, K \phi) \subseteq E_{\LY(s)}(\psi, \phi)
\end{align*}
for all $s \in (1/r, r]$.
\end{lemma}
\begin{remark}
The exceptional set in Equation~\eqref{eq:exceptional-set} was only defined for $x, y \in C \setminus \{0\}$, where $C$ is a proper complex cone. Above, we do \textit{not} require that $\phi \in \LY(r)$, although $\phi \neq 0$. This subtle distinction is important, as we will clarify below.
\end{remark}

\begin{proof}[Proof of Lemma~\ref{lem:ly-exceptional-set}]
For the first part of the lemma, the fact that $s \leq r$ and $\psi \in \LY(r)$ implies that $f_\psi(z) \neq 0$ for all $z \in \D_s^n$. With this in mind, if $\lambda \notin E_{\LY(s)}(\psi, \phi)$, then observe that $\lambda \psi - \phi \in \LY(s)$ if and only if $\lambda f_{\psi}(z) - f_{\phi}(z) \neq 0$ for any $z \in \D_s^n$. This, in turn, holds if and only if
\begin{align*}
    \lambda \notin \left\{\frac{f_\phi(z)}{f_\psi(z)} : z \in \D_s^n \right\}.
\end{align*}
The characterization of $E_{\LY(s)}(\psi, \phi)$ follows.

Next, because $K(\LY(s)) \subseteq \LY(r)$ for all $s \in (1/r, r]$, we have if $\lambda \notin E_{\LY(s)}(\psi, \phi)$, then the fact that $\lambda \psi - \phi \in \LY(s)$ again implies that
\begin{align*}
    \lambda K \psi - K \phi \in K(\LY(s)) \subseteq \LY(r).
\end{align*}
Thus, $\lambda \notin E_{\LY(r)}(K \psi, K \phi)$, proving the second part of the lemma.
\end{proof}

Equipped with this rather general lemma, we note the following general fact: if $\lambda_1, \lambda_2 \in \C^*$ and $C$ is a proper complex cone, then
\begin{align*}
    E_C(\lambda_1 x, \lambda_2 y) = \frac{\lambda_2}{\lambda_1} E_C(x, y)
\end{align*}
for all $x, y \in C \setminus \{0\}$. Even in the asymmetric formulation of Lemma~\ref{lem:ly-exceptional-set}, this property holds true. To see why, note that when $\zeta \notin E_C(\lambda_1 x, \lambda_2 y)$, then
\begin{align*}
    \zeta \lambda_1 x - \lambda_2 y \in C \implies \lambda_2 \left(\zeta \frac{\lambda_1}{\lambda_2} x - y\right) \in C \implies \zeta\frac{\lambda_1}{\lambda_2} x - y \in \frac{1}{\lambda_2} C \subseteq C.
\end{align*}
Hence, $\zeta\lambda_1 / \lambda_2 \notin E_C(x, y)$ so that $\zeta\notin \frac{\lambda_2}{\lambda_1} E_C(x, y)$. Conversely, if $\zeta\lambda_1 / \lambda_2 \notin E_C(x, y)$, then
\begin{align*}
    \zeta\frac{\lambda_1}{\lambda_2} x - y \in C \implies \zeta\lambda_1 x - \lambda_2 y \in \lambda_2 C \subseteq C \implies \zeta\notin E_C(\lambda_1 x, \lambda_2 y).
\end{align*}
This property of the exceptional set will prove particularly useful when $\phi$ and $\psi$ are eigenvectors of an operator $K$.

\subsection{The Squeezing Lemma and Diameter Bounds}
Taking $\phi$ and $\psi$ to be eigenvectors of $K$ with nonzero eigenvalues $\lambda_\phi$ and $\lambda_\psi$, we can recover the squeezing lemma of \cite{bravyi2026}.

\begin{lemma}[Generalized Squeezing Lemma]
\label{lem:gen-squeezing-lemma}
Suppose that $\phi$ and $\psi$ are linearly independent eigenvectors of $K$ with nonzero eigenvalues $\lambda_\phi$ and $\lambda_\psi$ and $\psi \in \LY(r)$. Then, for $s \in (1/r, r]$,
\begin{align*}
    \frac{\lambda_\phi}{\lambda_\psi} E_{\LY(r)}(\psi, \phi) \subseteq E_{\LY(s)}(\psi, \phi).
\end{align*}
\end{lemma}
\begin{proof}
Observe that
\begin{align*}
    \frac{\lambda_\phi}{\lambda_\psi} E_{\LY(r)}(\psi, \phi) = E_{\LY(r)}(\lambda_\psi \psi, \lambda_{\phi} \phi) &= E_{\LY(r)}(K \psi, K \phi) \\
    &\subseteq E_{\LY(s)}(\psi, \phi)
\end{align*}
as desired.
\end{proof}
In the notation of \cite{bravyi2026}, with
\begin{align}
\label{eq:eigenvector-g}
    g(z) := \frac{f_\phi(z)}{f_\psi(z)},
\end{align}
Lemma~\ref{lem:gen-squeezing-lemma} and the first part of Lemma~\ref{lem:ly-exceptional-set} combine to yield
\begin{align*}
    \frac{\lambda_\phi}{\lambda_\psi} g(\D_r^n) \subseteq g(\D_s^n)
\end{align*}
for any $s \in (1/r, r]$. Note that Lemma~\ref{lem:gen-squeezing-lemma} streamlines the diameter lower bound from \cite{bravyi2026}. Taking
\begin{align*}
    \Delta_s(g) := \diam(g(\overline{\D_s^n})),
\end{align*}
then for $1 / r < s < t < r$, we observe that
\begin{align*}
    \frac{\lambda_\phi}{\lambda_\psi} g(\overline{\D_t^n}) \subseteq \frac{\lambda_\phi}{\lambda_\psi} g(\D_r^n) \subseteq g(\D_s^n) \subseteq g(\overline{\D_s^n}),
\end{align*}
which immediately gives the lower bound
\begin{align}
\label{eq:diam-lower-bound}
\left|\frac{\lambda_\phi}{\lambda_\psi}\right| \Delta_t(g) \leq \Delta_s(g).
\end{align}
The diameter upper bound is then a multivariate generalization of the Landau--Toeplitz theorem. One formulation of this result can be found in \cite[Theorem 1.2]{Gentili2013} although the result was apparently proven in the same year as Schwarz's lemma \cite{Burckel2008}.
\begin{theorem}[Landau--Toeplitz Theorem]
\label{th:landau-toeplitz}
If $f : \D \to \C$ is holomorphic in $\D$ and $\diam(f(\D)) = 2$, then
\begin{align*}
    \diam(f(r \D)) \leq 2 r
\end{align*}
for all $r \in (0, 1)$ and $|f'(0)| \leq 1$.
\end{theorem}

We can recover the following, more general result that gives the diameter upper bound from \cite{bravyi2026}, with the exact same proof via the Schwarz lemma.

\begin{proposition}
\label{prop:diameter-upper-bound}
If $g : \D_r^n \to \C$ is a non-constant holomorphic function, in the sense that it is separately holomorphic in each variable, and $s \leq t < r$, then
\begin{align}
\label{eq:diameter-expansion}
    \diam(g(\overline{\D_s^n})) \leq \frac{s}{t} \diam(g(\overline{\D_t^n})).
\end{align}
\end{proposition}

It is possible to prove Proposition~\ref{prop:diameter-upper-bound} as a corollary of Theorem~\ref{th:landau-toeplitz}, but the argument in \cite{bravyi2026} is simpler and more direct. Lemma~\ref{lem:gen-squeezing-lemma} and Proposition~\ref{prop:diameter-upper-bound} combine to show that
\begin{align*}
    \left|\frac{\lambda_\phi}{\lambda_\psi}\right| \Delta_t(g) \leq \Delta_s(g) \leq \frac{s}{t} \Delta_t(g).
\end{align*}
for the choice of $g$ in Equation~\eqref{eq:eigenvector-g} and $1/r < s < t < r$.
Because $0 < \Delta_s(g) \leq \Delta_t(g) < \infty$ for any non-constant, holomorphic function $g$, taking $s \to 1/r$ and $t \to r$ as in \cite{bravyi2026} gives
\begin{align*}
    \left|\frac{\lambda_\phi}{\lambda_\psi}\right| \leq r^{-2}.
\end{align*}
\subsection{Extension to Non-Normal Operators}
The spectral ratio bound found in \cite[Theorem 2]{bravyi2026} establishes that for any positive semidefinite $n$-qubit operator $K \in \LY(r)$, the largest eigenvalue of $K$ is positive and non-degenerate with the spectral ratio bound
\begin{align*}
    \frac{\lambda_2}{\lambda_1} \leq r^{-2}.
\end{align*}
For non-normal matrices, the power iteration argument used to show that the dominant eigenvector $\psi$ of $K$ belongs to $\LY(r)$ breaks down. By carrying out an analogue of power iteration in projective space and applying Theorem~\ref{th:projective-contraction-principle}, we can entirely sidestep this obstacle and prove a spectral ratio for any operator that belongs to $\LY(r)$, regardless of whether it is normal or not. From this perspective, the operators in $\LY(r)$ are another broad class of operators that satisfy a Perron-Frobenius type theorem, as discussed in \cite{dubois2010complexifiedconesspectralgaps}.

Below, we assume that $r > 1$ and take $s \in (1/r, r)$. From Proposition~\ref{prop:ly-tensors-cone}, we know that $\LY(s) \subset V$ is a proper, linearly convex complex cone. Additionally, if $K \in \LY(r)$ is an operator, then $K(\LY(s)) \subseteq \LY(r) \subseteq \LY(s)$. Our first step is to show that the \textit{projective} diameter $\Delta_K$ appearing in Theorem~\ref{th:projective-contraction-principle} is finite so that $K$ is a strict contraction in projective space by way of two technical lemmas.
\begin{lemma}
\label{lem:bounded-oscillation}
Suppose $r > 1$ and $s \in (1/r, r)$, with $\psi \in \LY(r)$. Then, for any $z, w \in \D_s^n$, we have that
\begin{align}
\label{eq:mult-osc-bound}
\left|\frac{f_\psi(w)}{f_\psi(z)}\right| \leq \left(\frac{r+s}{r-s}\right)^n.
\end{align}
\end{lemma}
\begin{proof}
Write $z = (z_1, \ldots, z_n)$ and $w = (w_1, \ldots, w_n)$. For $i = 1, \ldots, n$, we define
\begin{align*}
    h_i(\zeta) := f_\psi(z_1,\ldots, z_{i-1}, \zeta, w_{i+1}, \ldots, w_n),
\end{align*}
noting that each $h_i$ is affine and non-vanishing on $\D_r$ because $f_\psi$ is multiaffine and belongs to $\LY(r)$. If $h_i$ is constant, then
\begin{align*}
    \left|\frac{h_i(w_i)}{h_i(z_i)}\right| \leq 1 \leq \frac{r + s}{r - s}.
\end{align*}
If $h_i$ is non-constant, then it admits a root $\lambda_i \in \C \setminus \D_r$, and we may write $h_i(\zeta) = c(\zeta - \lambda_i )$ for some $c \in \C^*$. In this case,
\begin{align*}
    \left|\frac{h_i(w_i)}{h_i(z_i)}\right| = \left|\frac{w_i - \lambda_i}{z_i - \lambda_i}\right| \leq \frac{|w_i| + |\lambda_i|}{|\lambda_i| - |z_i|} \leq \frac{s + |\lambda_i|}{|\lambda_i| - s} \leq \frac{r + s}{r - s}.
\end{align*}
To complete the proof, the fact that each of the $h_i$ is non-vanishing on $\D_r$ allows us to write
\begin{align*}
    f_\psi(w_1, \ldots, w_n) =  f_\psi(z_1, w_2, \ldots, w_n) \frac{h_1(w_1)}{h_1(z_1)} = \ldots = f_\psi(z_1, \ldots, z_n) \prod_{i=1}^n \frac{h_i(w_i)}{h_i(z_i)}
\end{align*}
so that
\begin{align*}
    \left|\frac{f_\psi(w)}{f_\psi(z)}\right| \leq \prod_{i=1}^n \left|\frac{h_i(w_i)}{h_i(z_i)}\right| \leq \left(\frac{r+s}{r-s}\right)^n
\end{align*}
as claimed.
\end{proof}

\begin{lemma}
\label{lem:K-projective-diameter}
If $r > 1$, $s \in (1/r, r)$, and $\delta_{\LY(s)}$ is the projective metric on $\LY(s)$ (cf. Equation~\eqref{eq:projective-metric}), then
\begin{align*}
    \Delta_K(s) := \sup_{\phi, \psi \in \LY(s)} \delta_{\LY(s)}(K\phi, K\psi) \leq \diam_{\delta_{\LY(s)}}(\LY(r)) \leq 2n \log\left(\frac{r + s}{r - s}\right),
\end{align*}
where
\begin{align*}
    \diam_{\delta_{\LY(s)}}(\LY(r)) := \sup_{\phi, \psi \in \LY(r)} \delta_{\LY(s)}(\phi,\psi).
\end{align*}
\end{lemma}
\begin{proof}
The quantity $\diam_{\delta_{\LY(s)}}(\LY(r)) \in [0, \infty]$ is well-defined (although not finite, a priori) because $\LY(r) \subseteq \LY(s)$. The first inequality then follows immediately because $K(\LY(s)) \subseteq \LY(r)$. Now, fix linearly independent $\phi, \psi \in \LY(r)$. From Lemma~\ref{lem:ly-exceptional-set}, we know that
\begin{align*}
    E_{\LY(s)}(\psi, \phi) = \left\{\frac{f_\phi(z)}{f_\psi(z)} : z \in \D_s^n \right\}.
\end{align*}
By Equation~\eqref{eq:projective-metric}, we can then compute the projective distance between $\phi$ and $\psi$:
\begin{align*}
    \delta_{\LY(s)}(\phi, \psi) &= \log\left(\sup_{z \in \D_s^n} \left|\frac{f_\phi(z)}{f_\psi(z)}\right|\right) - \log\left(\inf_{z \in \D_s^n} \left|\frac{f_\phi(z)}{f_\psi(z)}\right|\right) \\
    &\leq \sup_{z \in \D_s^n} \log\left|\frac{f_\phi(z)}{f_\psi(z)}\right| - \inf_{z \in \D_s^n} \log\left|\frac{f_\phi(z)}{f_\psi(z)}\right| \\
    &\leq (\sup_{z \in \D_s^n} \log|f_\psi(z)| - \inf_{z \in \D_s^n} \log|f_\psi(z)|) + (\sup_{z \in \D_s^n} \log|f_\phi(z)| - \inf_{z \in \D_s^n} \log|f_\phi(z)|).
\end{align*}
Now, for each $z, w \in \D_s^n$, Lemma~\ref{lem:bounded-oscillation} shows that
\begin{align*}
    \log|f_\psi(z)| -  \log|f_\psi(w)| \leq n \log\left(\frac{r + s}{r - s}\right).
\end{align*}
After taking supremum over $z, w \in \D_s^n$, this shows that
\begin{align*}
    \delta_{\LY(s)}(\phi, \psi) \leq 2n \log\left(\frac{r + s}{r - s}\right).
\end{align*}
The proof is then complete upon taking the supremum over all $\phi, \psi \in \LY(r)$.
\end{proof}

We now carry out a procedure akin to power iteration, but in projective space. The key inequality is Equation~\eqref{eq:contraction}, which in this setting, tells us that
\begin{align*}
    \delta_{\LY(s)}(K \psi, K \phi) \leq \tanh\left(\frac{n}{2} \log\left(\frac{r + s}{r - s}\right)\right) \delta_{\LY(s)}(\phi, \psi).
\end{align*}
Taking
\begin{align*}
    \rho := \tanh\left(\frac{n}{2} \log\left(\frac{r + s}{r - s}\right)\right) < 1,
\end{align*}
it follows that $K$ is a strict contraction in projective distance. Moreover, from Lemma~\ref{lem:projective-complete-metric-space}, we know that $(\LY(s) / {\sim}, \delta_{\LY(s)})$ is a complete metric space. To carry out this projective version of power iteration, we therefore fix $\psi_0 \in \LY(s)$ and define $\psi_k := K^k \psi_0$. Emulating the proof of the Banach fixed point theorem, it is easy to show that the equivalence classes of $\psi_k$ form a Cauchy sequence in $(\LY(s) / {\sim}, \delta_{\LY(s)})$ and hence converges. Recognizing that $K \psi \sim \psi$ implies that $\psi$ is an eigenvector of $K$ with non-zero eigenvalue, we obtain the following:
\begin{lemma}
\label{lem:projective-cauchy-sequence}
There exists $\psi \in \LY(r)$ satisfying $K \psi = \lambda_* \psi$ for some $\lambda_* \in \C^*$.
\end{lemma}
\begin{proof}
If $\pi : \LY(s) \to \LY(s) / {\sim}$ is the canonical projection, consider the map $\widehat{K} : \LY(s) /{\sim}\to \LY(s) / {\sim}$ given by
\begin{align*}
    \widehat{K}(\pi(\phi)) = \pi(K \phi).
\end{align*}
This is a well-defined map because if $\phi \in \LY(s)$ is such that $\phi \sim \phi'$, then $\phi' = c \phi$ for some $c \in \C^*$. Thus, $\phi' \in \LY(s)$ as well, and $\pi(K \phi') = \pi(c K \phi) = \pi(K \phi)$, noting that $K\phi, K \phi' \in \LY(s)$ too.

Now, $\delta_{\LY(s)}(K \psi, K \phi) \leq \rho\, \delta_{\LY(s)}(\phi, \psi)$ for all $\phi, \psi \in \LY(s)$, with $\rho < 1$. Because $f_\psi$ may vanish on the boundary of $\D_s^n$, the metric $\delta_{\LY(s)}$ may be infinite. As a result, instead of applying the Banach fixed point theorem directly, we replicate its proof with a few small modifications. Fix $\psi_0 \in \LY(s)$ and let $\psi_k := K^k \psi_0$. For $i \geq 1$,
\begin{align*}
    \delta_{\LY(s)}(\psi_{i+1}, \psi_i)
    \leq \rho^{i-1} \delta_{\LY(s)}(\psi_2, \psi_1)
    = \rho^{i-1} \delta_{\LY(s)}(K\psi_1, K\psi_0)
    \leq \rho^{i-1} \Delta_K(s),
\end{align*}
which is finite by Lemma~\ref{lem:K-projective-diameter}. Next, for $1 \leq j < k$,
\begin{align*}
    \delta_{\LY(s)}(\psi_j, \psi_k)
    \leq \sum_{i=j}^{k-1} \rho^{i-1} \Delta_K(s)
    \leq \frac{\rho^{j-1}}{1 - \rho} \Delta_K(s),
\end{align*}
so $(\pi(\psi_k))_{k \geq 1}$ is a Cauchy sequence in $(\LY(s)/{\sim}, \delta_{\LY(s)})$. By Lemma~\ref{lem:projective-complete-metric-space}, it converges to some $\pi(\psi)$ with $\psi \in \LY(s)$. In fact, $\pi(K\psi) = \pi(\psi)$. In other words, there exists $\lambda_* \in \C^*$ such that $K \psi = \lambda_* \psi$.
\end{proof}

We have shown that if $K \in \LY(r)$ is an $n$-qubit operator, then it admits an eigenvector $\psi \in \LY(r)$, with eigenvalue $\lambda_* \in \C^*$. We can now invoke the machinery that we developed above almost verbatim to prove Theorem~\ref{th:lee-yang-gap}, which holds for \textit{any} operator belonging to $\LY(r)$.
\begin{proof}[Proof of Theorem~\ref{th:lee-yang-gap}]
From Lemma~\ref{lem:projective-cauchy-sequence}, we know that there exists an eigenvector $\psi \in \LY(r)$ such that $K \psi = \lambda_* \psi$ for some $\lambda_* \in \C^*$. 

First, if $\phi$ is any other eigenvector of $K$ with nonzero eigenvalue $\lambda$ such that $\phi$ and $\psi$ are linearly independent, then Lemma~\ref{lem:gen-squeezing-lemma} reveals that
\begin{align*}
    \frac{\lambda}{\lambda_*} E_{\LY(r)}(\psi, \phi) \subseteq E_{\LY(s)}(\psi, \phi).
\end{align*}
Following Equation~\eqref{eq:diam-lower-bound}, we have that
\begin{align*}
    \left|\frac{\lambda}{\lambda_*}\right| \Delta_t(g) \leq \Delta_s(g)
\end{align*}
for all $1 / r < s < t < r$. On the other hand, Proposition~\ref{prop:diameter-upper-bound} gives
\begin{align*}
    \Delta_s(g) \leq \frac{s}{t} \Delta_t(g),
\end{align*}
and because $0 < \Delta_s(g) \leq \Delta_t(g) < \infty$ exactly as before, we may take $s \to 1/r$ and $t \to r$, obtaining
\begin{align*}
    \left|\frac{\lambda}{\lambda_*}\right| \leq r^{-2}
\end{align*}
for any other eigenvalue $\lambda$ of $K$. In particular, $\lambda_1 = \lambda_*$. Note that this shows that the geometric multiplicity of $\lambda_1$ is one because we would otherwise obtain $1 \leq r^{-2}$, which is impossible.

However, we must also rule out the presence of a generalized eigenvector belonging to the $\lambda_*$-generalized eigenspace. If such a generalized eigenvector $\phi$ were to exist, then we would have that
\begin{align*}
    K \phi = \lambda_* \phi + \psi
\end{align*}
due to the presence of a Jordan block of size at least two corresponding to $\lambda_*$. Now, for any $\zeta \in \C$, this implies that
\begin{align*}
    \zeta K \psi - K \phi = \lambda_*\left[\left(\zeta - \frac{1}{\lambda_*}\right)\psi - \phi\right],
\end{align*}
so that $\zeta \in E_{\LY(r)}(K \psi, K \phi)$ if and only if 
\begin{align*}
    \zeta - \frac{1}{\lambda_*} \in E_{\LY(r)}(\psi, \phi).
\end{align*}
In other words,
\begin{align*}
    \frac{1}{\lambda_*} + E_{\LY(r)}(\psi, \phi) = E_{\LY(r)}(K \psi, K \phi) \subseteq E_{\LY(s)}(\psi, \phi).
\end{align*}
The lefthand side contains $\frac{1}{\lambda_*} + g(\overline{\D_t^n})$ and the righthand side is contained in $g(\overline{\D_s^n})$, so because the diameter of a set is invariant under translations, we have that
\begin{align*}
    \Delta_t(g) \leq \Delta_s(g) \leq \frac{s}{t} \Delta_t(g)
\end{align*}
for all $1 / r < s < t < r$, which is impossible. In particular, no such Jordan block can exist, and we conclude that the algebraic multiplicity of $\lambda_*$ is one.
\end{proof}

\begin{remark}
As noted in \cite{bravyi2026}, the characterization of the spectral ratio provided in Theorem~\ref{th:lee-yang-gap} is sharp, even when restricted to positive semidefinite operators. As the most straightforward example of this tightness, consider the family of one-qubit diagonal operators
\begin{align*}
    K = \begin{pmatrix}
    1 & 0 \\ 0 & r^{-2}
    \end{pmatrix}
\end{align*}
for any $r > 1$.
Then, the spectral ratio of $K$ is exactly $r^{-2}$, and its generating polynomial is
\begin{align*}
    F_K(z, w) = 1 + r^{-2} z w,
\end{align*}
which is zero-free on $\D_r^2$ so that $K \in \LY(r)$.
\end{remark}

\subsection{A Local-to-Global Principle for Spectral Gaps}
\label{sec:local-to-global}
Because the spectral ratio result from Theorem~\ref{th:lee-yang-gap} holds for non-normal operators, Trotterization suddenly becomes a much more useful tool. To demonstrate this new technique, we focus on the case of EPR-like Hamiltonians from \cite{wong2026}. However, the technique that follows may be useful as a local-to-global principle for establishing gaps via Theorem~\ref{th:lee-yang-gap} and Trotterization. The outline of the strategy is quite simple:
\begin{enumerate}
\item[(1)] Establish that each local term in the Hamiltonian has Lee-Yang radius greater than one.
\item[(2)] Choose a Trotterization of the Hamiltonian such that the Lee-Yang radius of the resulting product remains greater than one.
\item[(3)] Use Theorem~\ref{th:lee-yang-gap} to establish the spectral ratio of the Trotterization itself and pass to the limit. 
\end{enumerate}
Without Theorem~\ref{th:lee-yang-gap}, it is not possible to say anything about the spectral ratio of the Trotterization itself, although Steps (1) and (2) of the above procedure were already accessible using the tools from \cite{wong2026} and \cite{bravyi2026}.

To see this strategy in action, we recall that the EPR-like Hamiltonian on an arbitrary connected interaction graph $G = (V, E)$ is of the following form:
\begin{align*}
    H_s := - \sum_{(i, j) \in E} \ket{\phi_s} \bra{\phi_s}_{ij}, \quad \ket{\phi_s} = \ket{00} + s \ket{11}.
\end{align*}
For simplicity below, we may write each local term in the Hamiltonian as
\begin{align*}
    h_{ij}^{(s)} := -\ket{\phi_s} \bra{\phi_s}_{ij}.
\end{align*}

The only other key fact that we will need is the following result from \cite{wong2026}, concerning the EPR-like Hamiltonian on a single edge.

\begin{lemma}[Lemma 11 in \cite{wong2026}]
\label{lem:epr-like-edge-radius}
Suppose that $n = 2$ and that $G$ is the graph on a single edge. Then, $\exp(-\beta H_s) \in \LY(r_\beta(s))$, where
\begin{align*}
    r_\beta(s) = \left(\frac{1 + s^2 e^{-\beta(1 + s^2)}}{s^2 + e^{-\beta(1 + s^2)}}\right)^{1/4}.
\end{align*}
\end{lemma}

We will prove the following result about EPR-like Hamiltonians by way of the strategy described above:
\begin{theorem}
\label{th:epr-gap}
Suppose that $H_s$ is the EPR-like Hamiltonian on a connected interaction graph $G = (V, E)$, with parameter $s \in [0, 1)$. Then, the spectral gap of $H_s$, denoted by $\gamma(H_s)$, satisfies
\begin{align}
\label{eq:epr-gap}
\gamma(H_s) \geq \frac{1 - s^2}{2}.
\end{align}
\end{theorem}

Below, it will be useful to note that for small $\beta$ and $s \in [0, 1)$, we can write
\begin{align*}
    \log r_\beta(s) = \frac{1 - s^2}{4} \beta + O(\beta^2),
\end{align*}
a computation carried out in detail in Appendix~\ref{app:computations}.

With this in mind, we write the Gibbs state on an arbitrary interaction graph via the Lie product formula as
\begin{align}
\label{eq:trotterization}
    \exp(-\beta H_s) &= \exp\left(-\beta \sum_{(i, j) \in E} h_{ij}^{(s)}\right) \notag \\
    &= \lim_{M \to \infty} \left(\prod_{(i, j) \in E} \exp\left(-\frac{\beta}{M} h_{ij}^{(s)}\right)\right)^M,
\end{align}
with the convergence occurring in operator norm. For an operator on $n$ qubits, the following proof will require slightly more granular information about its Lee-Yang radius. In particular, following the convention of \cite{wong2026}, we can view such an operator as a tensor with $2n$ indices, enumerated as $(i_1, \ldots, i_n, i_1', \ldots, i_n')$. We sometimes refer to the first $n$ indices as ``ingoing'' indices, while the last $n$ may be called ``outgoing'' indices.

Now, Lemma~\ref{lem:epr-like-edge-radius} tells us that each local term
\begin{align}
\label{eq:full-ly-indices}
\exp\left(-\frac{\beta}{M} h_{ij}^{(s)}\right) \in \LY(1, \ldots, \underbrace{r_{\beta/M}(s)}_{\text{slot } i}, \ldots,
  \underbrace{r_{\beta/M}(s)}_{\text{slot } j}, \ldots, 1, 1, \ldots, \underbrace{r_{\beta/M}(s)}_{\text{slot } i'}, \ldots,
  \underbrace{r_{\beta/M}(s)}_{\text{slot } j'}, \ldots, 1).
\end{align}
To simplify notation below, we suppress all indices in Equation~\eqref{eq:full-ly-indices} that are equal to one, writing
\begin{align*}
\exp\left(-\frac{\beta}{M} h_{ij}^{(s)}\right) \in \LY(\underbrace{r_{\beta/M}(s)}_{\text{slot } i},
\underbrace{r_{\beta/M}(s)}_{\text{slot } j},\underbrace{r_{\beta/M}(s)}_{\text{slot } i'}, \underbrace{r_{\beta/M}(s)}_{\text{slot } j'}),
\end{align*}
for example.
Moreover, note that the identity operator on $n$ qubits admits the generating polynomial
\begin{align*}
    F_I(z, w) = \prod_{i=1}^n (1 + z_i w_i).
\end{align*}
Although it is certainly the case that $I \in \LY(1, \ldots, 1)$, it is often more useful to note that for any $r > 0$, we have that
\begin{align*}
    I \in \LY(1, \ldots, \underbrace{r}_{\text{slot } i}, \ldots, 1, 1, \ldots, \underbrace{r^{-1}}_{\text{slot } i'}, \ldots, 1) \cap \LY(1, \ldots, \underbrace{r^{-1}}_{\text{slot } i}, \ldots, 1, 1, \ldots, \underbrace{r}_{\text{slot } i'}, \ldots, 1)
\end{align*}

Using this, we can inductively produce a Lee-Yangness guarantee for the whole product that appears in the Trotterization.
\begin{lemma}
\label{lem:trotterized-ly}
For all $M \in \N$ and any ordering of local terms in the Trotterization in Equation~\eqref{eq:trotterization},  we have that
\begin{align*}
     \left(\prod_{(i, j) \in E} \exp\left(-\frac{\beta}{M} h_{ij}^{(s)}\right)\right)^M \in \LY(r_{\beta / M}(s)).
\end{align*}
\end{lemma}
\begin{proof}
Suppose that the first two local terms that appear in the Trotterized product correspond to edge $(i, j)$ and edge $(k, \ell)$, respectively. There are two possible cases that we must consider. First, if $\{i, j\} \cap \{k, \ell\} = \emptyset$, then in the notation of Equation~\eqref{eq:full-ly-indices}, indices $i, j, k, \ell$ and $i', j', k', \ell'$ are all occupied by $r_{\beta / M}(s)$ upon taking the product of the two operators. Indeed,
\begin{align*}
    \exp\left(-\frac{\beta}{M} h_{ij}^{(s)}\right) \exp\left(-\frac{\beta}{M} h_{k \ell}^{(s)}\right)
\end{align*}
can be expressed as the tensor product of the two local terms, with identity matrices occupying the four remaining vertices.

On the other hand, if $|\{i,j\} \cap \{k,\ell\}| = 1$, assume without loss of generality that $i = k$ and that $i = k < j < \ell$. Using the observation about the identity matrix made above, we note that
\begin{align*}
\exp\left(-\frac{\beta}{M} h_{ij}^{(s)}\right) \in \LY(\underbrace{r_{\beta/M}(s)}_{\text{slot } i},
\underbrace{r_{\beta/M}(s)}_{\text{slot } j}, \underbrace{r_{\beta/M}(s)}_{\text{slot } \ell}, \underbrace{r_{\beta/M}(s)}_{\text{slot } i'}, \underbrace{r_{\beta/M}(s)}_{\text{slot } j'}, \underbrace{r_{\beta/M}(s)^{-1}}_{\text{slot } \ell'}),
\end{align*}
and
\begin{align*}
\exp\left(-\frac{\beta}{M} h_{k \ell}^{(s)}\right) \in \LY(\underbrace{r_{\beta/M}(s)}_{\text{slot } i},
\underbrace{r_{\beta/M}(s)^{-1}}_{\text{slot } j}, \underbrace{r_{\beta/M}(s)}_{\text{slot } \ell}, \underbrace{r_{\beta/M}(s)}_{\text{slot } i'}, \underbrace{r_{\beta/M}(s)}_{\text{slot } j'}, \underbrace{r_{\beta/M}(s)}_{\text{slot } \ell'}),
\end{align*}
Recalling from \cite[Lemma 1]{wong2026} the closure properties of the cone of Lee-Yang tensors under tensor contraction, it follows that either the resulting product is zero or
\begin{align*}
\exp\left(-\frac{\beta}{M} h_{ij}^{(s)}\right) \exp\left(-\frac{\beta}{M} h_{k \ell}^{(s)}\right) \in \LY(\underbrace{r_{\beta/M}(s)}_{\text{slot } i},
\underbrace{r_{\beta/M}(s)}_{\text{slot } j}, \underbrace{r_{\beta/M}(s)}_{\text{slot } \ell}, \underbrace{r_{\beta/M}(s)}_{\text{slot } i'}, \underbrace{r_{\beta/M}(s)}_{\text{slot } j'}, \underbrace{r_{\beta/M}(s)}_{\text{slot } \ell'}). 
\end{align*}
The product of the two local terms cannot be zero because all matrix exponentials involved are invertible and hence have nonzero product. Thus, the second option must occur.
In either case, for all vertices involved in the first two terms in the Trotterization, the Lee-Yang radius along the corresponding indices is given by $r_{\beta / M}(s)$.

We repeat this procedure inductively: at step $m + 1$ of the induction, suppose that local terms on vertices $\{i_1, \ldots, i_{n(m)}\}$ have already appeared in the Trotterization, where $n(m) \leq 2m$. If the next local term involves vertices $\{j, k\}$, then one of \textit{three} cases can now occur. If $\{j, k\} \cap \{i_1, \ldots, i_{n(m)}\} = \emptyset$, then the same reasoning as above ensures that all Lee-Yang radii along indices $\{i_1, \ldots, i_{n(m)}\} \cup \{j,k\}$ are $r_{\beta / M}(s)$ after evaluating the product of the first $m + 1$ terms in the Trotterization.

If $|\{j, k\} \cap \{i_1, \ldots, i_{n(m)}\}| = 1$ or $\{j, k\} \subset \{i_1, \ldots, i_{n(m)}\}$, then to all indices in $\{i_1, \ldots, i_{n(m)}\} \setminus \{j, k\}$, we associate an ingoing Lee-Yang radius of $r_{\beta / M}(s)^{-1}$ and an outgoing Lee-Yang radius of $r_{\beta / M}(s)$. When $|\{j, k\} \cap \{i_1, \ldots, i_{n(m)}\}| = 1$, then supposing that $j \notin \{i_1, \ldots, i_{n(m)}\}$, we must also associate an ingoing Lee-Yang radius of $r_{\beta / M}(s)$ and an outgoing Lee-Yang radius of $r_{\beta / M}(s)^{-1}$ with index $j$ in the product of the first $m$ terms of the Trotterization (which would otherwise still have radius one). After carrying out this bookkeeping, we observe that the Lee-Yang radius along all ingoing \textit{and} outgoing indices in $\{i_1, \ldots, i_{n(m)}\} \cup \{j, k\}$ is $r_{\beta / M}(s)$ upon taking the product of the first $m + 1$ terms in the Trotterization.

To conclude the proof, we note that because the interaction graph is connected, all vertices appear in at least one edge. Hence, upon carrying out the above inductive procedure until it terminates, we have that
\begin{align*}
    \prod_{(i, j) \in E} \exp\left(-\frac{\beta}{M} h_{ij}^{(s)}\right) \in \LY(r_{\beta / M}(s))
\end{align*}
regardless of the ordering of the edges in the Trotterization. It then immediately follows that
\begin{align*}
     \left(\prod_{(i, j) \in E} \exp\left(-\frac{\beta}{M} h_{ij}^{(s)}\right)\right)^M \in \LY(r_{\beta / M}(s)).
\end{align*}
because integer powers of Lee-Yang operators remain in the Lee-Yang cone of the same radius, again invoking \cite[Lemma 1]{wong2026}.
\end{proof}
At this point, it is tempting to simply take the limit as $M \to \infty$. However, because $r_{\beta / M}(s) \to 1$ as $M \to \infty$, taking such a limit only tells us that $\exp(-\beta H_s) \in \LY(1)$, which is not sufficient to obtain a spectral gap. Rather, we apply Theorem~\ref{th:lee-yang-gap} \textit{prior} to taking the $M \to \infty$ limit.

\begin{proof}[Proof of Theorem~\ref{th:epr-gap}]
Fixing $\beta \in (0, \infty)$ and letting
\begin{align*}
    T_M := \prod_{(i, j) \in E} \exp\left(-\frac{\beta}{M} h_{ij}^{(s)}\right),
\end{align*}
Theorem~\ref{th:lee-yang-gap} implies that
\begin{align*}
    \left|\frac{\lambda_2(T_M)}{\lambda_1(T_M)}\right| \leq r_{\beta / M}(s)^{-2},
\end{align*}
with the eigenvalues of $T_M$ ordered as usual. In turn, it follows that
\begin{align*}
    \left|\frac{\lambda_2(T_M^M)}{\lambda_1(T_M^M)}\right| \leq r_{\beta / M}(s)^{-2 M}.
\end{align*}
Because $T_M^M \to \exp(-\beta H_s)$ in operator norm,
\begin{align*}
    \log \left|\frac{\lambda_2(\exp(-\beta H_s))}{\lambda_1(\exp(-\beta H_s))}\right| &\leq \lim_{M \to \infty} [-2 M \log r_{\beta / M}(s)] \\
    &= \lim_{M \to \infty} \left[-\left(\frac{1 - s^2}{2}\right)\beta + O (\beta^2 / M) \right] \\
    &= -\left(\frac{1 - s^2}{2}\right)\beta.
\end{align*}
Writing the lefthand side in terms of the spectral gap $\gamma(H_s)$ of the EPR-like Hamiltonian, note that
\begin{align*}
    \frac{\lambda_2(\exp(-\beta H_s))}{\lambda_1(\exp(-\beta H_s))} = \exp(-\beta \gamma(H_s)),
\end{align*}
In particular, taking logarithms and rearranging gives
\begin{align*}
    \gamma(H_s) \geq \frac{1 - s^2}{2}
\end{align*}
as claimed.
\end{proof}

This conclusion is also an immediate corollary of \cite[Theorem 1]{bravyi2026}, although we hope that the local-to-global principle presented in this example may find use elsewhere. In particular, we only required that each local term has Lee-Yang radius $r_\beta(s)$ satisfying
\begin{align*}
    \log r_\beta(s) = c \beta + O(\beta^2)
\end{align*}
for some $c > 0$ as $\beta \to 0$. As the results in the following section demonstrate, this method of bounding spectral gaps cannot be extended beyond the setting of Suzuki-Fisher Hamiltonians. However, for Suzuki-Fisher Hamiltonians with zero external field, the methods in \cite{bravyi2026} do not apply. Hence, it still may be the case that the technique used to obtain Theorem~\ref{th:epr-gap} may yield useful information about the gaps of such Hamiltonians, with the caveat that many such models (e.g., the ferromagnetic XXZ chain) have a degenerate ground state \cite{koma1997}.

\section{Characterizing Lee-Yang Tensors}
\label{sec:ly-characterization}
In this section, we describe the limitations of the framework for Lee-Yangness introduced in \cite{wong2026}. First, Section~\ref{sec:k-local} demonstrates that $k$-local Hamiltonians cannot satisfy $\exp(-\beta H) \in \LY(1)$ for all inverse temperatures $\beta \geq 0$, even under arbitrary local changes-of-basis (see Definition~\ref{def:local-unitary}). Along the way, we provide a powerful yet elementary perturbative method for proving that the unnormalized Gibbs state $\exp(-\beta H)$ of a Hamiltonian $H$ does \textit{not} belong to $\LY(1)$ for sufficiently small $\beta > 0$. In particular, we apply the same perturbative proof technique in Section~\ref{sec:ly-to-sf} to show that if a Hamiltonian $H$ satisfies $\exp(-\beta H) \in \LY(1)$ for all $\beta \geq 0$, then $H$ is a Suzuki-Fisher Hamiltonian. Hence, Lee-Yangness at all temperatures is \textit{equivalent} to membership in the Suzuki-Fisher class of Hamiltonians.

See the remark in the introduction for a comparison of our work with that of \cite{wei2026}, which recovers many of the same results using entirely different machinery.

\subsection{$k$-Local Hamiltonians}
\label{sec:k-local}
Following the results presented above and the work in \cite{bravyi2026}, it is natural to ask whether the framework of Lee-Yangness can be fruitfully applied beyond the class of Suzuki-Fisher Hamiltonians. One might hope that certain $k$-local Hamiltonians, in which each local term involves at most $k$ qubits, benefit from a result akin to Theorem~\ref{th:lee-yang-gap}. Unfortunately, this doesn't appear to be the case: for any $k \geq 3$, we can exhibit a gapped $k$-local Hamiltonian $H$ such that the Gibbs state $\exp(-\beta H) \notin \LY(1)$ for all sufficiently small $\beta > 0$. By extension, the local-to-global principle presented in the previous section, which relies on a Lee-Yang radius strictly greater than one at \textit{all} sufficiently small positive inverse temperatures, is no longer applicable.

We use a similar notion of $k$-locality as in \cite{Harrow2020}, although our definition extends to arbitrary interaction graphs and requires the presence of a term of support of size exactly $k$.

\begin{definition}
\label{def:k-local}
A Hamiltonian on an interaction graph $G = (V, E)$ is $k$-local if it can be written in the form
\begin{align*}
    H = \sum_{\substack{X \subseteq V \\ |X| \leq k}} H_X
\end{align*}
for Hermitian operators $H_X$ that act non-trivially on sites in $X \subseteq V$.
\end{definition}

Below, it will be useful to express Hamiltonians in terms of the Pauli basis. First, let $n := |V|$, the number of qubits. For each multi-index $\balpha =(\alpha_1, \ldots, \alpha_n) \in \{I, x, y, z\}^n$, we take
\begin{align*}
    P_{\balpha} := \bigotimes_{i = 1}^n \sigma_{\balpha_i},
\end{align*}
where $\sigma_{\alpha}$ are the standard Pauli matrices
\begin{align*}
    \sigma_I = \begin{pmatrix}
    1 & 0 \\ 0 & 1 
    \end{pmatrix}, \quad \sigma_x = \begin{pmatrix}
    0 & 1 \\ 1 & 0
    \end{pmatrix}, \quad \sigma_y = \begin{pmatrix}
    0 & -i \\ i & 0
    \end{pmatrix}, \quad \sigma_z = \begin{pmatrix}
    1 &  0 \\ 0 & -1
    \end{pmatrix}
\end{align*}
for $\alpha \in \{I, x,y,z\}$. If we write the nontrivial support of a multi-index $\balpha \in \{I, x,y,z\}^n$ as
\begin{align*}
    \mathrm{supp}(\balpha) := \{i \in [n] : \alpha_i \in \{x,y,z\}\},
\end{align*}
then any $k$-local Hamiltonian $H$ can be expressed
\begin{align}
\label{eq:pauli-form}
    H = \sum_{\substack{\balpha \in \{I, x,y,z\}^n \\ |\mathrm{supp}(\balpha)| \leq k}} h_{\balpha} P_{\balpha}
\end{align}
where $h_{\balpha}$ are real coefficients. We say that $H$ is \textit{genuinely} $k$-local if there exists $\balpha \in \{I,x,y,z\}^n$ with $|\supp(\balpha)| = k$ and $h_{\balpha} \neq 0$.

Before stating our main result for $k$-local Hamiltonians, we also define a local unitary operator. The restriction to such unitaries is necessary because arbitrary entangled changes-of-basis can alter the locality of the Hamiltonian in question.
\begin{definition}
\label{def:local-unitary}
A unitary operator $U$ on $n$ qubits is \textit{local} if it tensorizes as $U = U_1 \otimes \ldots \otimes U_n$, where $U_i$ is a unitary operator on a single qubit for each $i \in [n]$.
\end{definition}

Equipped with these preliminaries, we prove that Gibbs states of genuinely $k$-local Hamiltonians fail to belong to $\LY(1)$ at sufficiently high temperature regardless of the choice of local basis.

\begin{theorem}
\label{th:k-local-obstruction}
Suppose that $H$ is a genuinely $k$-local Hamiltonian on $n$ qubits, with $k \geq 3$. Then, there exists an inverse temperature $\beta(H) > 0$ such that for all $\beta \in (0, \beta(H))$ and any local unitary $U$,
\begin{align*}
    U^* \exp(-\beta H) U \notin \LY(1).
\end{align*}
\end{theorem}

To prove this result, we carry out a perturbative analysis of the generating polynomial of the identity matrix. Recall that
\begin{align*}
    F_I(z, w) = \prod_{i=1}^n (1 + z_i w_i),
\end{align*}
which admits roots on the boundary of the unit polydisk. Moving towards this analysis, observe that for any unitary operator $U$, we have that $U^* \exp(-\beta H) U = \exp(-\beta U^* H U)$. To simplify notation, let $H(U) := U^* H U$, the image of the $k$-local Hamiltonian $H$ under conjugation by a local unitary $U$. 
Our general strategy proceeds by defining
\begin{align}
    M_\beta(U) := \exp(-\beta H(U)) - I + \beta H(U)
\end{align}
and writing
\begin{align*}
    U^* \exp(-\beta H) U = \exp(-\beta H(U)) = I - \beta H(U) + M_\beta(U).
\end{align*}
At the level of generating polynomials,
\begin{align*}
    F_{U^* \exp(-\beta H) U}(z, w) = F_I(z, w) - \beta F_{H(U)}(z, w) + F_{M_\beta}(z, w).
\end{align*}

As noted above, the first term has roots of multiplicity $k$ on the boundary of the unit polydisk. For sufficiently small $\beta$, the latter two terms perturb these roots symmetrically, ensuring that at least one enters the interior of the polydisk, as depicted in Figure~\ref{fig:root-shattering}.

\begin{figure}[h]
    \centering
    \includegraphics[width=0.8\linewidth]{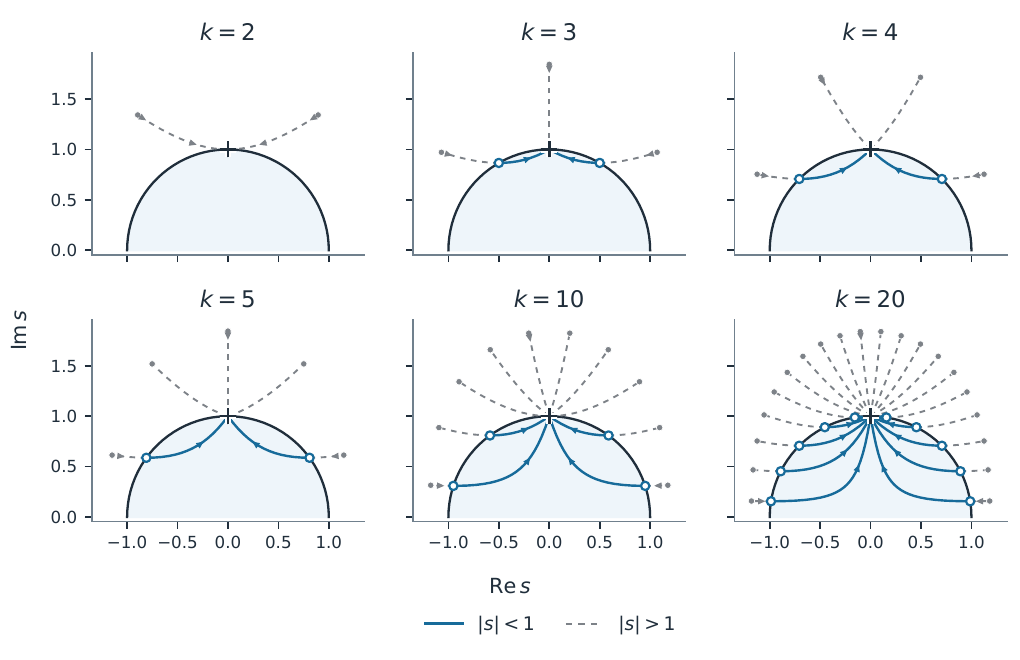}
    \caption{The roots of the univariate polynomial $F_{\exp(-\beta H)}(s, \ldots, s)$ for a $k$-local Hamiltonian of the form $H = I - \ket{\psi} \bra{\psi}$, where $\ket{\psi} = \ket{0}^{\otimes k}$. For $k = 2$, the resulting Hamiltonian is Suzuki-Fisher and hence its roots avoid the interior of the unit disk, but for all $k \geq 3$, roots shatter symmetrically from $s = i$ for sufficiently small $\beta > 0$, ensuring that at least one enters the interior of the disk.}
    \label{fig:root-shattering}
\end{figure}

Our perturbative analysis essentially proceeds by reducing to the study of the roots of a univariate polynomial of the form
\begin{align}
\label{eq:toy-example}
    f_\beta(s) = [2(1 - s)]^k - \beta C,
\end{align}
where $C \neq 0$ and $\beta \geq 0$. For $\beta > 0$, the $k$ complex roots of this univariate polynomial may be enumerated as
\begin{align*}
    r_m = 1 - \frac{\beta^{1/k} c_m}{2}, \quad m = 1, \ldots, k,
\end{align*}
where $c_1, \ldots, c_k$ are the $k$th roots of $C$. For $k \geq 3$, the fact that the arguments of these $k$ roots are equispaced on the unit circle ensures that, for some $m \in [k]$, $\Re(c_m) > 0$. Up to relabeling, we may assume that $\Re(c_1) > 0$. In turn, for sufficiently small $\beta > 0$, we have that
\begin{align*}
    |r_1|^2 = 1 - \beta^{1/k} \Re(c_1) + \frac{\beta^{2/k} |c_1|^2}{4}.
\end{align*}
As soon as
\begin{align}
\label{eq:beta-threshold}
    \beta < \left(\frac{\Re(c_1)}{4 |c_1|^2}\right)^k,
\end{align}
it therefore follows that $|r_1| < 1$. In other words, for sufficiently small $\beta$, at least one root of $f_\beta(s)$ must lie inside the unit disk. Crucially, we used that $\Re(c_1) > 0$ above to give a threshold below which a zero enters the disk. Below, we use a classical tool from complex analysis, Rouché's theorem, to compare the roots of an appropriate univariate specialization of the generating polynomial of a $k$-local Hamiltonian to a univariate polynomial like the one in Equation~\eqref{eq:toy-example}.

We assume, without loss of generality, that given a genuinely $k$-local Hamiltonian $H$, there is a local term $P_{\balpha}$ acting non-trivially on the qubits at sites $\{1, \ldots, k\}$, where we have enumerated the vertices $V = \{1, \ldots, n\}$ of the underlying interaction graph. Given this convention, we can always select a point on the boundary of the unit polydisk for which $F_I(z, w) = 0$ but $F_{H(U)}(z, w)$ is nonzero. Selecting such a point, which depends on the Hamiltonian $H$ and the local unitary $U$ in question, will give us access to the nonzero constant $C \neq 0$ that appears in Equation~\eqref{eq:toy-example}. The proof, via basic Fourier analysis, can be found in Appendix~\ref{app:proofs}.

\begin{lemma}
\label{lem:boundary-root}
For any genuinely $k$-local Hamiltonian $H$ and any local unitary operator $U$, there exists a point of the form $(z_\theta(U), w_\theta(U)) = (e^{i \theta_1(U)}, \ldots, e^{i \theta_k(U)}, 0, \ldots, 0, -e^{-i \theta_1(U)}, \ldots, -e^{-i \theta_k(U)}, 0, \ldots, 0)$ such that $F_{H(U)}(z_\theta(U), w_\theta(U)) \neq 0$. Here, $\theta(U) = (\theta_1(U), \ldots, \theta_k(U)) \in [0, 2\pi)^k$ depends on $U$.
\end{lemma}

We now fix a local unitary $U$ and suppress the dependence of all quantities on $U$. In particular, we fix the point $(z_\theta, w_\theta) = (z_\theta(U), w_\theta(U))$ furnished by Lemma~\ref{lem:boundary-root} and simply write $H$ in place of $H(U)$. To proceed, we expand the operator exponential into its linearization plus an error term of the form $M_\beta := (\exp(-\beta H) - I + \beta H)$, as described above:
\begin{align}
\label{eq:exp-decomp}
\exp(-\beta H) = I - \beta H + M_\beta.
\end{align}
A coarse upper bound on $M_\beta$, in operator norm, can be given as follows:
\begin{align*}
    \|M_\beta\| = \left\|\sum_{m=2}^\infty \frac{(-\beta H)^m}{m!}\right\| \leq \frac{\beta^2 \|H\|^2}{2} \exp(\beta \|H\|).
\end{align*}
Hence, we anticipate that, even at the level of generating polynomials, the error term in Equation~\eqref{eq:exp-decomp} decays as $O(\beta^2)$ so if $F_{I - \beta H}(z, w)$ has a root inside the unit polydisk, then so does $F_{\exp(-\beta H)}(z, w)$. With this intuition in mind, we narrow our focus to the study of the roots of
\begin{align*}
    P_\beta(s) := F_{I - \beta H}(s z_\theta, s w_\theta) &= F_I(s z_\theta, s w_\theta) - \beta F_H(s z_\theta, s w_\theta) \\
    &= (1 - s^2)^k - \beta F_H(s z_\theta, s w_\theta).
\end{align*}
To complete our reduction to the simple example presented in Equation~\eqref{eq:toy-example}, note that near $s = 1$, we anticipate that $(1 - s^2)^k \approx [2(1 - s)]^k$ while $F_H(s z_\theta, s w_\theta) \neq 0$. Hence, we may instead study the roots of
\begin{align*}
    Q_\beta(s) := [2(1 - s)]^k - \beta C(H),
\end{align*}
where $C(H) := F_H(z_\theta, w_\theta) \neq 0$ from Lemma~\ref{lem:boundary-root}. We already know that $Q_\beta(s)$ admits a root inside the unit disk for sufficiently small $\beta$. We now apply Rouché's theorem, recalled below, twice: first, to show that $P_\beta(s)$ has a root inside the disk by comparison to $Q_\beta(s)$, and second, to show that $F_{\exp(-\beta H)}(s z_\theta, s w_\theta)$ also has a root inside the disk by comparison to $P_\beta(s)$.

\begin{fact}[Rouché's Theorem]
\label{fact:rouche}
Suppose that $f$ and $g$ are holomorphic on an open set containing a closed disk $D$. If $|f(s)| > |g(s)|$ for all $s \in \partial D$, then $f$ and $f + g$ have the same number of zeros in the interior of $D$.
\end{fact}

Taking $c_1$ to be a $k$th root of $C(H)$ with positive real part and
\begin{align*}
    r_1(\beta) = 1 - \frac{\beta^{1/k} c_1}{2}
\end{align*}
the corresponding root of $Q_\beta(s)$ as above, we know that for $\beta$ as in Equation~\eqref{eq:beta-threshold}, $r_1(\beta) \in \D$. In particular, for $\varepsilon > 0$ sufficiently small (to be determined later), the ball $B_\varepsilon(r_1(\beta))$ of radius $\varepsilon$ about $r_1(\beta)$ is entirely enclosed in the unit disk and contains a single root of $Q_\beta(s)$. Letting $D_\varepsilon(r_1(\beta)) := \overline{B_\varepsilon(r_1(\beta))}$, we can carry out the first comparison, yielding a root of $P_\beta(s)$ inside the unit disk.

\begin{lemma}
\label{lem:linearization-root}
For all sufficiently small $\beta > 0$, the univariate polynomial $P_\beta(s)$ has a single root inside $D_\varepsilon(r_1(\beta)) \subset \D$, where
\begin{align*}
    \varepsilon = \frac{\beta^{1/k}}{4} \min\left\{\Re(c_1), |c_1 - c_2|, \ldots, |c_1 - c_k|\right\}.
\end{align*}
\end{lemma}
\begin{proof}
As indicated above, we simply apply Rouché's theorem with $f(s) = Q_\beta(s)$ and $g(s) = P_\beta(s) - Q_\beta(s)$. First, we verify that $D_\varepsilon(r_1(\beta)) \subset \D$ under our choice of $\varepsilon$. Indeed, for all $s \in D_\varepsilon(r_1(\beta))$, we have that
\begin{align*}
    \Re(1 - s) \geq \frac{\beta^{1/k} \Re(c_1)}{2} - \varepsilon \geq \frac{\beta^{1/k} \Re(c_1)}{4},
\end{align*}
while
\begin{align*}
    |1 - s| \leq |s - r_1(\beta)| + |r_1(\beta) - 1| \leq \varepsilon + \frac{\beta^{1/k}|c_1|}{2} \leq \beta^{1/k} |c_1|.
\end{align*}
In particular, we may bound
\begin{align*}
    |s|^2 = 1 - 2\Re(1 - s) + |1 - s|^2 \leq 1 - \frac{\beta^{1/k} \Re(c_1)}{4} +\beta^{2/k} |c_1|^2 < 1
\end{align*}
for all $\beta > 0$ satisfying Equation~\eqref{eq:beta-threshold}. Additionally, our choice of $\varepsilon$ ensures that $Q_\beta(s)$ has no roots other than $r_1(\beta)$ inside $D_\varepsilon(r_1(\beta))$, as
\begin{align*}
    |r_m(\beta) - r_1(\beta)| = \frac{\beta^{1/k}}{2}|c_1 - c_m| > \varepsilon
\end{align*}
for each $m = 2, \ldots, k$.

Now, we may fully factorize $Q_\beta(s)$ in terms of its roots, writing
\begin{align*}
    Q_\beta(s) = (-2)^k \prod_{m=1}^k (s - r_m(\beta)).
\end{align*}
On $\partial D_\varepsilon(r_1(\beta))$, observe that for $m = 2, \ldots, k$
\begin{align*}
    |s - r_m(\beta)| \geq |r_1(\beta) - r_m(\beta)| - \varepsilon \geq \frac{\beta^{1/k} |c_1 - c_m|}{4}.
\end{align*}
Hence, on $\partial D_\varepsilon(r_1(\beta))$,
\begin{align*}
    |Q_\beta(s)| = 2^k \prod_{m=1}^k |s - r_m(\beta)| \geq c \beta,
\end{align*}
where $c > 0$ is a constant that depends only on the roots of $C(H)$ (and, of course, $k$). On the other hand, we have that
\begin{align*}
    P_\beta(s) - Q_\beta(s) = (1 - s)^k ((1 + s)^k - 2^k) - \beta[F_H(s z_\theta, s w_\theta) - F_H(z_\theta, w_\theta)].
\end{align*}
The univariate polynomial $s \mapsto F_H(s z_\theta, s w_\theta)$ is Lipschitz on compact domains (e.g., the closed disk of radius two about the origin), so there exists another constant $L = L(H) > 0$, depending only on $H$, such that
\begin{align*}
    |F_H(s z_\theta, s w_\theta) - F_H(z_\theta, w_\theta)| \leq L|1 - s| \leq L |c_1| \beta^{1/k}.
\end{align*}
Taking $L > 0$ larger if necessary, the same reasoning applied to the map $s \mapsto (1 + s)^k$ reveals that
\begin{align*}
    |P_\beta(s) - Q_\beta(s)| \leq L |1 - s|^{k+1} + \beta^{1 + 1/k}L |c_1| \leq L(|c_1| + |c_1|^{k+1}) \beta^{1 + 1/k}.
\end{align*}
Crucially, $\beta^{1 + 1/k}$ decays faster than $\beta$ as $\beta \to 0$, so as long as
\begin{align*}
    \beta < \left(\frac{c}{L(|c_1| + |c_1|^{k+1})}\right)^k,
\end{align*}
a threshold that depends only on $H$, we have that
\begin{align}
\label{eq:rouche-estimate-1}
|Q_\beta(s)| \geq c \beta > L(|c_1| + |c_1|^{k+1}) \beta^{1 + 1/k} \geq |P_\beta(s) - Q_\beta(s)|
\end{align}
for all $s \in \partial D_\varepsilon(r_1(\beta))$. By Fact~\ref{fact:rouche}, we immediately conclude that $P_\beta(s)$ has precisely one root inside the disk $D_\varepsilon(r_1(\beta)) \subset \D$.
\end{proof}

Having established that $P_\beta(s)$, a univariate specialization of the generating polynomial of $I - \beta H$, has at least one root inside the unit disk (so that $I - \beta H \notin \LY(1)$ for sufficiently small $\beta > 0$), we can carry out a similar strategy to extend this result to $F_{\exp(-\beta H)}(sz_\theta, sw_\theta)$ itself.

\begin{lemma}
\label{lem:specialization-root}
Taking $r_1(\beta)$ and $\varepsilon > 0$ as in Lemma~\ref{lem:linearization-root}, the univariate polynomial
\begin{align*}
    F_{\exp(-\beta H)}(s z_\theta, s w_\theta) = P_\beta(s) + F_{M_\beta}(s z_\theta, s w_\theta)
\end{align*}
has precisely one root inside the disk $D_\varepsilon(r_1(\beta)) \subset \D$ for all sufficiently small $\beta > 0$.
\end{lemma}
\begin{proof}
We again apply Fact~\ref{fact:rouche}, this time taking $f(s) = P_\beta(s)$ and $g(s) = F_{M_\beta}(s z_\theta, s w_\theta)$. Now, a straightforward computation (see Appendix~\ref{app:computations}) shows that for all $s \in \partial D_\varepsilon(r_1(\beta)) \subset \D$
\begin{align*}
    |g(s)| \leq 2^k \|M_\beta\| \leq (2^{k - 1} \|H\|^2 \exp(\beta \|H\|)) \beta^2.
\end{align*}
If $\beta < \|H\|^{-1}$, we can remove the dependence of the coefficient on $\beta$ and bound $|g(s)| < 2^{k+1} \|H\|^{2} \beta^2$ for $s \in \partial D_\varepsilon(r_1(\beta))$.

Using the computations from the proof of Lemma~\ref{lem:linearization-root}, we can also give a lower bound on $P_\beta(s)$ on $\partial D_\varepsilon(r_1(\beta))$ by noting that
\begin{align*}
    |P_\beta(s)| \geq |Q_\beta(s)| - |P_\beta(s) - Q_\beta(s)| \geq c\beta - L(|c_1| + |c_1|^{k+1}) \beta^{1 + 1/k},
\end{align*}
with $c, L > 0$ constants that depend only on $H$. Because the upper bound on $|g(s)|$ decays as $O(\beta^2)$, while the lower bound on $|f(s)|$ decays as $\Omega(\beta)$, it is clear that for sufficiently small $\beta > 0$, the conditions of Rouché's theorem will be met. More explicitly, a tedious algebraic manipulation shows that if we further take
\begin{align*}
    \beta < \min\left\{\|H\|^{-1}, \left(\frac{c}{2 L(|c_1| + |c_1|^{k+1})}\right)^k, \frac{c}{2^{k+2} \|H\|^2}\right\},
\end{align*}
then $|P_\beta(s)| > |F_{M_\beta}(s z_\theta, s w_\theta)|$ for all $s \in \partial D_\varepsilon(r_1(\beta))$. In turn, we may apply Fact~\ref{fact:rouche} to conclude that 
\begin{align*}
    F_{\exp(-\beta H)}(s z_\theta, s w_\theta) = P_\beta(s) + F_{M_\beta}(s z_\theta, s w_\theta)
\end{align*}
has precisely one zero in the disk $D_\varepsilon(r_1(\beta)) \subset \D$.
\end{proof}

As indicated in the statement of Theorem~\ref{th:k-local-obstruction}, the inverse temperature threshold $\beta(H)$ below which $U^* \exp(-\beta H) U \notin \LY(1)$ is \textit{uniform} across all local unitaries $U$. We are now equipped to prove said theorem by a compactness argument.

\begin{proof}[Proof of Theorem~\ref{th:k-local-obstruction}]
The set of local unitaries $\sL := \{U_1 \otimes \ldots \otimes U_n : U_i \in U(2)\}$ is a closed subset of $U(2^n)$ and hence compact. As a result, we proceed by covering $\sL$ by a finite collection of neighborhoods and showing that on each neighborhood, we have uniform control over the threshold afforded by Lemma~\ref{lem:specialization-root}.

To this end, fix a unitary $U \in \sL$ and choose $(z_\theta, w_\theta) = (z_\theta(U), w_\theta(U))$ according to Lemma~\ref{lem:boundary-root} so that $F_{H(U)}(z_\theta(U), w_\theta(U)) \neq 0$. Additionally, let $\varepsilon = \varepsilon(U) > 0$ be the radius and $\beta(U) > 0$ the inverse temperature threshold from Lemma~\ref{lem:specialization-root}, both depending on $U$. 

If $U' \in \sL$ is another local unitary, then we have that
\begin{align*}
    \|\exp(-\beta H(U)) - \exp(-\beta H(U'))\| \leq 2 \beta \|H\| \exp(\beta \|H\|) \|U - U'\|.
\end{align*}
In turn, for all $|s| \leq 1$, we show in Appendix~\ref{app:computations} that
\begin{align*}
    |F_{\exp(-\beta H(U))}(s z_\theta, s w_\theta) - F_{\exp(-\beta H(U'))}(s z_\theta, s w_\theta)| &\leq 2^{k} \|\exp(-\beta H(U)) - \exp(-\beta H(U'))\| \\
    &\leq 2^{k+1} \beta \|H\| \exp(\beta \|H\|) \|U - U'\|.
\end{align*}
On the other hand, our estimate from Lemma~\ref{lem:specialization-root} shows that if $\beta < \beta(U)$, then for $s \in \partial D_\varepsilon(r_1(\beta))$,
\begin{align*}
    |F_{\exp(-\beta H(U))}(s z_\theta, s w_\theta)| \geq c(U) \beta - c'(U) \beta^{1 + 1/k} > \frac{c(U)}{2}\beta
\end{align*}
where $c(U), c'(U) > 0$ are constants depending on $U$. If necessary, we may decrease $\beta(U)$ so that this bound holds for all $\beta < \beta(U)$. Hence, as long as 
\begin{align}
\label{eq:neighborhood-radius}
    \|U - U'\| < \frac{c(U)}{2^{k+4}\|H\|},
\end{align}
we can ensure that
\begin{align*}
    |F_{\exp(-\beta H(U))}(s z_\theta, s w_\theta) - F_{\exp(-\beta H(U'))}(s z_\theta, s w_\theta)| \leq \frac{c(U)}{2} \beta < |F_{\exp(-\beta H(U))}(s z_\theta, s w_\theta)|
\end{align*}
for $s \in \partial D_\varepsilon(r_1(\beta))$. Applying Fact~\ref{fact:rouche} yet again with $f(s) = F_{\exp(-\beta H(U))}(s z_\theta, s w_\theta)$ and 
\begin{align*}
    g(s) = F_{\exp(-\beta H(U'))}(s z_\theta, s w_\theta) - F_{\exp(-\beta H(U))}(s z_\theta, s w_\theta)
\end{align*}
then shows that $F_{\exp(-\beta H(U'))}(s z_\theta, s w_\theta)$ has a root inside $D_\varepsilon(r_1(\beta)) \subset \D$. 

In less opaque language, we have shown that for all $U' \in \sL$ in the neighborhood of $U$ specified by Equation~\eqref{eq:neighborhood-radius}, we have that $(U')^* \exp(-\beta H) U' \notin \LY(1)$ for all $\beta \in (0, \beta(U))$. As noted above, $\sL$ is compact. Thus, covering $\sL$ by a finite collection of such neighborhoods as in Equation~\eqref{eq:neighborhood-radius}, each centered at a local unitary $U^{(1)}, \ldots, U^{(N)}$, we may then take
\begin{align*}
    \beta(H) := \min_{i = 1, \ldots, N} \beta(U^{(i)}),
\end{align*}
completing the proof.
\end{proof}

\subsection{Lee-Yangness at All Temperatures and Suzuki-Fisher Hamiltonians}
\label{sec:ly-to-sf}
When this paper was nearing completion, the author noticed the work of \cite{wei2026}, a wide-ranging preprint that recovers many of the results above, albeit using quite different techniques. Part of \cite[Theorem 5.2]{wei2026} stands out in particular: if $\exp(-\beta H) \in \LY(1)$ for all $\beta$, then $H$ is a Suzuki-Fisher Hamiltonian. This apparent converse of \cite[Theorem 2]{Suzuki1971} reveals that the local-to-global principle in Section~\ref{sec:local-to-global} and the method in \cite[Theorem 1]{bravyi2026} are \textit{only} applicable to Suzuki-Fisher Hamiltonians. Although \cite[Theorem 5.2]{wei2026} does not rule out Lee-Yangness at low temperatures (e.g., as $\beta \to \infty$), the framework is fundamentally limited to the Suzuki-Fisher class when it comes to obtaining spectral gaps via passing to the $\beta \to 0$ limit, either via Trotterization as in our work or by adding an external field as in \cite{bravyi2026}.

Upon reading about this surprising result in \cite{wei2026}, we realized that the perturbative analysis in Section~\ref{sec:k-local} can provide a simple proof of the following fact.

\begin{theorem}
\label{th:lee-yang-converse}
A Hamiltonian $H$ satisfies $\exp(-\beta H) \in \LY(1)$ for all $\beta > 0$ if and only if $H$ is a Suzuki-Fisher Hamiltonian in the sense that, in the Pauli basis,
\begin{align*}
    H = cI -\sum_{\{i,j\} \in E} (J_{ij}^{zz} \sigma_{z}^i \sigma_z^j + J_{ij}^{xx}  \sigma_{x}^i \sigma_x^j + J_{ij}^{yy}  \sigma_{y}^i \sigma_y^j + J_{ij}^{xy}  \sigma_{x}^i \sigma_y^j + J_{ij}^{yx}  \sigma_{y}^i \sigma_x^j) - \sum_{i \in V} (\mu_i^x \sigma_x^i + \mu_i^y \sigma_y^i + \mu_i^z \sigma_z^i),
\end{align*} 
with $c \in \R$, $\mu_i^z \geq 0$ for all $i \in V$, and
\begin{align}
\label{eq:sf-couplings}
    J_{ij}^{zz} \geq \left\|\begin{pmatrix}
    J_{ij}^{xx} & J_{ij}^{xy} \\ J_{ij}^{yx} & J_{ij}^{yy}
    \end{pmatrix}\right\|
\end{align}
for all $\{i, j\} \in E$.
\end{theorem}

As a side note, observe that if $\exp(-\beta H) \in \LY(1)$ for all $\beta \in (0, \beta(H))$, then $\exp(-\beta H) \in \LY(1)$ for all $\beta > 0$ by observing that $\exp(-\beta H / M)^M = \exp(-\beta H)$ for any $M > 0$.
Moreover, from Section~\ref{sec:k-local}, we already know that $H$ cannot have genuinely $k$-local interactions for $k \geq 3$ if $\exp(-\beta H) \in \LY(1)$ for all $\beta > 0$, so we can immediately reduce to studying $2$-local Hamiltonians of the form
\begin{align*}
    H = cI - \sum_{\{i, j\} \in E} \sum_{\alpha, \gamma \in \{x, y, z\}} J_{ij}^{\alpha \gamma} \sigma_{\alpha}^i \sigma_{\gamma}^j - \sum_{i \in V} \sum_{\alpha \in \{x, y, z\}} \mu_i^\alpha \sigma_{\alpha}^i,
\end{align*}
where $c$ is some real constant and $G = (V, E)$ is a connected interaction graph. Indeed, any $2$-local Hamiltonian takes the above form in the Pauli basis. Surprisingly, the perturbative technique from the previous section is fully equipped to handle Hamiltonians of this form. As before, we will reduce to analyzing the zeros of the univariate polynomials $s \mapsto F_H(s z_\theta, s w_\theta)$, where $(z_\theta, w_\theta)$ is a point on the boundary of the unit polydisk, parametrized as in Lemma~\ref{lem:boundary-root}. The mechanism of proof is essentially the same as in Section~\ref{sec:k-local}. If $F_H(z_\theta, w_\theta) \notin (-\infty, 0]$ for some point $(z_\theta, w_\theta)$ of the form
\begin{align}
\label{eq:2-local-boundary-points}
    z_\theta = (0, \ldots, 0, e^{i \theta_1}, 0, \ldots, 0, e^{i \theta_2}, 0, \ldots 0), \quad w_\theta = (0, \ldots, 0, -e^{-i \theta_1}, 0, \ldots, 0, -e^{-i \theta_2}, 0, \ldots 0),
\end{align}
we will show that $\exp(-\beta H) \notin \LY(1)$ at sufficiently high temperature by essentially the same perturbative method in Section~\ref{sec:k-local}. Specifically, as in Section~\ref{sec:k-local}, the points $(z_\theta, w_\theta)$ are precisely those on the boundary of the unit polydisk for which the generating polynomial of the identity vanishes. When $F_H(z_\theta, w_\theta) \notin (-\infty, 0]$, we are able to choose a square root of $F_H(z_\theta, w_\theta)$ with positive real part, at which point the same argument via Rouché's theorem that we used to prove Theorem~\ref{th:k-local-obstruction} applies. On the other hand, we show that if $F_H(z_\theta, w_\theta) \in (-\infty, 0]$ for all $(\theta_1, \theta_2) \in [0, 2\pi)^2$, then $H$ is a Suzuki-Fisher Hamiltonian. Following this outline, we relegate the bulk of Theorem~\ref{th:lee-yang-converse} into a technical lemma, proven in Appendix~\ref{app:proofs}.

\begin{lemma}
\label{lem:ly-to-sf}
If $H$ is a $2$-local Hamiltonian and $\exp(-\beta H) \in \LY(1)$ for all $\beta \geq 0$, then $H$ is a Suzuki-Fisher Hamiltonian.
\end{lemma}

With this result in hand, we can give a simple proof of Theorem~\ref{th:lee-yang-converse}.

\begin{proof}[Proof of Theorem~\ref{th:lee-yang-converse}]
If $\exp(-\beta H) \in \LY(1)$ for all $\beta \geq 0$, then Theorem~\ref{th:k-local-obstruction} implies that $H$ is $2$-local. Hence, by Lemma~\ref{lem:ly-to-sf}, $H$ is a Suzuki-Fisher Hamiltonian. Conversely, if $H$ is a Suzuki-Fisher Hamiltonian, then the classical result of \cite[Theorem 2]{Suzuki1971} (see also \cite[Fact 1]{wong2026} for a more modern statement) implies that $\exp(-\beta H) \in \LY(1)$ for all $\beta \geq 0$.
\end{proof}

\printbibliography

\appendix
\section{Deferred Proofs}
\label{app:proofs}
In this section, we provide proofs of several of the technical lemmas presented in Section~\ref{sec:k-local} and Section~\ref{sec:ly-to-sf}. For all of the proofs, the following fact, a direct result of the definition of the generating polynomial of a single-qubit operator, will be useful to keep in mind.

\begin{fact}
\label{fact:pauli-gen-poly}
If $\sigma_{x}, \sigma_y,  \sigma_z$ are the Pauli matrices, then
\begin{align*}
    F_{\sigma_x}(e^{i \theta}, -e^{-i \theta}) = 2 i \sin(\theta), \quad F_{\sigma_y}(e^{i \theta}, -e^{-i \theta}) = 2 i \cos(\theta), \quad F_{\sigma_z}(e^{i \theta}, -e^{-i \theta}) = 2.
\end{align*}
\end{fact}
In particular, if $\balpha = (\alpha_1, \ldots, \alpha_n) \in \{I, x,y,z\}^n$ and
\begin{align*}
    P_{\balpha} = \bigotimes_{i=1}^n \sigma_{\alpha_i},
\end{align*}
then
\begin{align*}
    F_{P_{\balpha}}(z, w) = \prod_{i=1}^n F_{\sigma_{\alpha_i}}(z_i, w_i).
\end{align*}
If $U = U_1 \otimes \ldots \otimes U_n$ is a local unitary and
\begin{align*}
    H = \sum_{\substack{\balpha \in \{I, x,y,z\}^n \\ |\mathrm{supp}(\balpha)| \leq k}} h_{\balpha} P_{\balpha},
\end{align*}
then we may write
\begin{align}
\label{eq:H(U)-pauli}
    H(U) := U^* H U = \sum_{\substack{\balpha \in \{I, x,y,z\}^n \\ |\mathrm{supp}(\balpha)| \leq k}} h_{\balpha}(U) P_{\balpha}.
\end{align}

\subsection*{From Section~\ref{sec:k-local}}
The only remaining loose end in Section~\ref{sec:k-local} is Lemma~\ref{lem:boundary-root}. Recall from said section that we assumed that $H$ is a genuinely $k$-local term containing a term $P_{\balpha}$ acting non-trivially on the first $k$ qubits, having enumerated the underlying interaction graph as $V = \{1, \ldots, n\}$. Because conjugation by a local unitary $U = U_1 \otimes \ldots \otimes U_n$ is trace preserving, it follows that $U^* P_{\balpha} U$ also acts non-trivially on the first $k$ qubits and trivially (i.e., as the identity) on the remaining $n - k$. Specifically, $\Tr(U_i^* \sigma_{\alpha_i} U_i) = \Tr(\sigma_{\alpha_i}) = 0$ for $\alpha_i \in \{x,y,z\}$, so the support of $U^* P_{\balpha} U$ is the same as the support of $P_{\balpha}$.

Now, if $(\theta_1, \ldots, \theta_k) \in [0, 2\pi)^k$, then for points of the form
\begin{align*}
    (z_\theta, w_\theta) = (e^{i \theta_1}, \ldots, e^{i \theta_k}, 0, \ldots, 0, -e^{-i \theta_1}, \ldots, -e^{-i \theta_k}, 0, \ldots, 0),
\end{align*}
observe that
\begin{align}
\label{eq:boundary-eval-i}
    F_{H(U)}(z_\theta, w_\theta) = \sum_{\substack{\balpha \in \{I, x, y, z\}^n \\ \supp(\balpha) = \{1, \ldots, k\}}} h_{\balpha}(U) F_{P_{\balpha}}(z_\theta, w_\theta),
\end{align}
Indeed, if a term corresponding $\balpha \in \{I, x, y, z\}^n$ acts trivially on a site $i \in \{1,\ldots,k\}$, then the corresponding generating polynomial vanishes because $F_I(e^{i \theta}, -e^{-i \theta}) = 0$ for all $\theta \in [0, 2\pi)$. We also apply the fact that $F_I(0, 0) = 1$ to obtain Equation~\eqref{eq:boundary-eval-i}. Define now
\begin{align*}
    g_U(\theta_1, \ldots, \theta_k) := \sum_{\substack{\balpha \in \{I, x, y, z\}^n \\ \supp(\balpha) = \{1, \ldots, k\}}} h_{\balpha}(U) F_{P_{\balpha}}(z_\theta, w_\theta),
\end{align*}
which we may view as an element of $L^2(\T^k)$, where $\T$ is the segment $[0, 2 \pi]$ with the endpoints identified. We will ultimately show that $\|g_U\|_{L^2(\T^k)}^2 > 0$, at which point Lemma~\ref{lem:boundary-root} follows because there then must exist a point $\theta = (\theta_1, \ldots, \theta_k) \in \T^k$ such that
\begin{align*}
    F_{H(U)}(z_\theta, w_\theta) = g_U(\theta_1, \ldots, \theta_k) \neq 0,
\end{align*}
where $(z_\theta, w_\theta)$ is fully determined by the point $\theta = (\theta_1, \ldots, \theta_k) \in \T^k$. First, it is helpful to note that
\begin{align*}
    F_{P_{\balpha}}(z_\theta, w_\theta) = \prod_{j=1}^k F_{\sigma_{\alpha_j}}(e^{i \theta_j}, -e^{-i\theta_j}).
\end{align*}
By Fact~\ref{fact:pauli-gen-poly}, each of the terms in the above product is a trigonometric function in $\theta_j$. In turn, $F_{P_{\balpha}}(z_\theta, w_\theta)$ is a trigonometric polynomial in the variables $(\theta_1, \ldots, \theta_k)$. The three trigonometric functions that arise in Fact~\ref{fact:pauli-gen-poly} are moreover orthogonal to each other as elements of $L^2(\T)$. Writing $g_x(\theta) = 2 i \sin(\theta)$, $g_y(\theta) = 2i \cos(\theta)$, and $g_z(\theta) = 2$, we may compute
\begin{align*}
    \langle g_x, g_y \rangle_{L^2(\T)} = \frac{1}{2\pi} \int_0^{2\pi} (-2i \sin(\theta))(2i \cos(\theta)) d\theta = 0,
\end{align*}
and the remaining orthogonality relations can be checked directly for $\alpha, \gamma \in \{x,y,z\}$ with $\alpha \neq \gamma$. Although the functions in Fact~\ref{fact:pauli-gen-poly} are not normalized, direct computation shows that $\|g_\alpha\|_{L^2(\T)}^2 \geq 2$ for all $\alpha \in \{x,y,z\}$.

\begin{proof}[Proof of Lemma~\ref{lem:boundary-root}]
As described above, it suffices to show that $\|g_U\|_{L^2(\T^k)}^2 > 0$. Denoting
\begin{align*}
    g_{\balpha}(\theta_1, \ldots, \theta_k) = \prod_{j=1}^k g_{\alpha_j}(\theta_j),
\end{align*}
we may directly expand
\begin{align*}
    \|g_U\|_{L^2(\T^k)}^2 = \sum_{\substack{\balpha, \bm{\beta} \in \{I, x,y,z\}^n \\ \mathrm{supp}(\balpha) = \mathrm{supp}(\bm{\beta}) = \{1, \ldots, k\}}} h_{\balpha}(U)^* h_{\bm{\beta}}(U) \langle g_{\balpha}, g_{\bm{\beta}}\rangle_{L^2(\T^k)}.
\end{align*}
For $\balpha \neq \bm{\beta}$, the fact that $\langle g_\alpha, g_\gamma \rangle_{L^2(\T)} = 0$ for $\alpha \neq \gamma \in \{x,y,z\}$ yields
\begin{align*}
    \|g_U\|_{L^2(\T^k)}^2 = \sum_{\substack{\balpha \in \{I, x,y,z\}^n \\ \mathrm{supp}(\balpha) = \{1, \ldots, k\}}} |h_{\balpha}(U)|^2 \prod_{j=1}^k \|g_{\alpha_j}\|_{L^2(\T)}^2 \geq 2^k \sum_{\substack{\balpha \in \{I, x,y,z\}^n \\ \mathrm{supp}(\balpha) = \{1, \ldots, k\}}} |h_{\balpha}(U)|^2.
\end{align*}
Because $H$ is a genuinely $k$-local Hamiltonian with at least one term acting non-trivially on $\{1, \ldots, k\}$ in the Pauli basis, we know that
\begin{align*}
    \sum_{\substack{\balpha \in \{I, x,y,z\}^n \\ \mathrm{supp}(\balpha) = \{1, \ldots, k\}}} |h_{\balpha}|^2 > 0.
\end{align*}
Hence, we can conclude the proof upon observing that
\begin{align*}
    \sum_{\substack{\balpha \in \{I, x,y,z\}^n \\ \mathrm{supp}(\balpha) = \{1, \ldots, k\}}} |h_{\balpha}(U)|^2  = \sum_{\substack{\balpha \in \{I, x,y,z\}^n \\ \mathrm{supp}(\balpha) = \{1, \ldots, k\}}} |h_{\balpha}|^2.
\end{align*}
To see this, note that as long as $U$ is a local unitary, we have that
\begin{align*}
    \sum_{\substack{\balpha \in \{I, x,y,z\}^n \\ \mathrm{supp}(\balpha) = \{1, \ldots, k\}}} h_{\balpha}(U) P_{\balpha} = U^* \left(\sum_{\substack{\balpha \in \{I, x,y,z\}^n \\ \mathrm{supp}(\balpha) = \{1, \ldots, k\}}} h_{\balpha} P_{\balpha} \right) U,
\end{align*}
so taking the Hilbert-Schmidt norm of both sides proves the identity and thus the lemma.
\end{proof}

\subsection*{From Section~\ref{sec:ly-to-sf}}

Our first step towards proving Theorem~\ref{th:lee-yang-converse}, by way of Lemma~\ref{lem:ly-to-sf} is the following fact that closely mirrors our analysis of $k$-local Hamiltonians. In the $k$-local case, we crucially used the fact that if $k \geq 3$, then any nonzero constant admits a $k$th root with positive real part. By finding a zero $(z_\theta, w_\theta)$ of the identity's generating polynomial on the boundary of the unit polydisk (via Lemma~\ref{lem:boundary-root}) such that $F_H(z_\theta, w_\theta) \neq 0$, we were able to show that $F_{\exp(-\beta H)}(s z_\theta, sw_\theta)$ vanishes for some $|s| < 1$, as long as $\beta$ is sufficiently small. 

In the $2$-local case, our strategy is essentially identical, with the caveat that we now require access to a point $(z_\theta, w_\theta)$ on the boundary of the unit polydisk such that $F_H(z_\theta, w_\theta) \notin (-\infty, 0] \subset \C$. Indeed, for any $C \notin (-\infty, 0]$, we are able to pick a square root of $C$ with positive real part. After this, the proofs of Lemma~\ref{lem:linearization-root} and Lemma~\ref{lem:specialization-root} both remain valid in this setting, mutatis mutandis (i.e., taking $k = 2$ throughout), allowing us to conclude the following.

\begin{lemma}
\label{lem:2-local-perturbative}
If $H$ is a $2$-local Hamiltonian and there exists a point as in Equation~\eqref{eq:2-local-boundary-points} such that $F_H(z_\theta, w_\theta) \notin (-\infty, 0]$, then for some $\beta(H) > 0$, $\exp(-\beta H) \notin \LY(1)$ for all $\beta \in (0, \beta(H))$.
\end{lemma}

We can now prove Lemma~\ref{lem:ly-to-sf}. The conditions on the couplings follow from requiring that $F_H(z_\theta, w_\theta) \in (-\infty, 0]$ at every point of the form in Equation~\eqref{eq:2-local-boundary-points}, which Lemma~\ref{lem:2-local-perturbative} forces whenever $\exp(-\beta H) \in \LY(1)$ for all $\beta \geq 0$. The condition that $\mu_i^z \geq 0$ for all $i \in V$ is slightly harder to check, so we separate it into its own lemma. Throughout, we write $2$-local Hamiltonians in the Pauli basis as
\begin{align}
\label{eq:generic-2-local}
    H = cI - \sum_{\{i, j\} \in E} \sum_{\alpha, \gamma \in \{x, y, z\}} J_{ij}^{\alpha \gamma} \sigma_{\alpha}^i \sigma_{\gamma}^j - \sum_{i \in V} \sum_{\alpha \in \{x, y, z\}} \mu_i^\alpha \sigma_{\alpha}^i.
\end{align}
We also require a simple lemma about positive definite Lee-Yang operators on a single qubit.

\begin{lemma}
\label{lem:single-qubit-inequality}
Suppose that $A \in \LY(1)$ is a positive definite operator. If we write
\begin{align*}
    A = \begin{pmatrix}
    a_{00} & a_{01} \\ a_{01}^* & a_{11}
    \end{pmatrix}
\end{align*}
for $a_{00}, a_{11} > 0$, then $a_{00} \geq a_{11}$.
\end{lemma}
\begin{proof}
Expand the generating polynomial of $A$ as
\begin{align*}
    F_A(z, w) = a_{00} + a_{01} w + a_{01}^* z + a_{11} zw.
\end{align*}
If $a_{01} = 0$, then taking $(z, w) = (s, -s)$ means that the real-rooted univariate polynomial $f(s) = a_{00} - a_{11}s^2$ has no roots inside the unit disk, in which case $a_{00} \geq a_{11}$.

If $a_{01} \neq 0$, then consider points of the form $(z, w) = (s e^{i \theta}, -se^{-i\theta})$, where $\theta = \arg(a_{01})$. At such points, we again require that $f(s) = a_{00} - a_{11}s^2$ has no roots inside the unit disk, which recovers the inequality $a_{00} \geq a_{11}$.
\end{proof}

In the following lemma, $\Tr_{\Lambda}(A)$ denotes the \textit{normalized} partial trace of an operator over a subset of sites $\Lambda \subset V$, so that $\Tr_{\Lambda}(I) = I_{V \setminus \Lambda}$.

\begin{lemma}
\label{lem:ly-to-positive-field}
If $H$ is a $2$-local Hamiltonian in the Pauli basis such that $\exp(-\beta H) \in \LY(1)$ for all $\beta \geq 0$, then $\mu_i^z \geq 0$ for each $i \in V$.
\end{lemma}
\begin{proof}
If $\exp(-\beta H) \in \LY(1)$ for all $\beta \geq 0$, then by \cite[Lemma 1]{wong2026}, all partial traces of $\exp(-\beta H)$ also belong to $\LY(1)$ for all $\beta \geq 0$. In particular, if we trace out every site other than a fixed site $i \in V$, it follows that
\begin{align*}
    T_\beta^i = \Tr_{V \setminus \{i\}}(\exp(-\beta H)) \in \LY(1),
\end{align*}
as long as it is not zero, and because $\exp(-\beta H)$ is positive definite, it cannot be the case that $T_\beta^i = 0$. Now, observe that for $2$-local Hamiltonians that take the form in Equation~\eqref{eq:generic-2-local},
\begin{align*}
    \Tr_{{V \setminus \{i\}}}(H) = cI - \mu_i^x \sigma_x - \mu_i^y \sigma_y - \mu_i^z \sigma_z \in \C^{2 \times 2},
\end{align*}
as the partial trace of any term containing a non-identity Pauli matrix on a site $j \neq i$ is zero. Writing $\exp(-\beta H) = I - \beta H + (\exp(-\beta H) - I + \beta H)$ as before, we compute
\begin{align*}
    \Tr_{V \setminus \{i\}}(\exp(-\beta H)) = I - \beta(cI - \mu_i^x \sigma_x - \mu_i^y \sigma_y - \mu_i^z \sigma_z) + \Tr_{V \setminus \{i\}}(\exp(-\beta H) - I + \beta H).
\end{align*}
The normalized trace of the remainder term $R_\beta^i := \Tr_{V \setminus \{i\}}(\exp(-\beta H) - I + \beta H)$ satisfies
\begin{align*}
    \|R_\beta^i\| = \|\Tr_{V \setminus \{i\}}(\exp(-\beta H) - I + \beta H)\| \leq \|\exp(-\beta H) - I + \beta H\| \leq \frac{1}{2} \beta^2 \|H\|^2 \exp(\beta \|H\|).
\end{align*}
Hence, we have shown that $T_\beta^i \in \LY(1)$ for all $\beta \geq 0$ and written
\begin{align*}
    T_\beta^i = I - \beta(cI - \mu_i^x \sigma_x - \mu_i^y \sigma_y - \mu_i^z \sigma_z) + R_\beta^i,
\end{align*}
with the norm of the remainder term $R_\beta^i$ decaying as $O(\beta^2)$. From Lemma~\ref{lem:single-qubit-inequality}, we have that
\begin{align*}
    1 - \beta c + \beta \mu_i^z + (R_\beta^i)_{00} \geq 1 - \beta c - \beta \mu_i^z + (R_\beta^i)_{11} \implies 2 \beta \mu_i^z + (R_\beta^i)_{00} - (R_\beta^i)_{11} \geq 0.
\end{align*}
Upon bounding
\begin{align*}
    |(R_\beta^i)_{00} - (R_\beta^i)_{11}| \leq 2 \|R_\beta^i\| \leq \beta^2\|H\|^2 \exp(\beta \|H\|)
\end{align*}
and dividing through by $\beta > 0$, it follows that
\begin{align*}
    2 \mu_i^z + \beta \|H\|^2 \exp(\beta \|H\|) \geq 0
\end{align*}
for all $\beta > 0$. Taking the limit $\beta \to 0$ then implies that $\mu_i^z \geq 0$.
\end{proof}

Having shown that any $2$-local Hamiltonian such that $\exp(-\beta H) \in \LY(1)$ for all $\beta \geq 0$ must have $\mu_i^z \geq 0$ for all $i \in V$, it remains to check that the coupling constants $J_{ij}^{\alpha \gamma}$ satisfy Equation~\eqref{eq:sf-couplings} for all $\{i, j\} \in E$. We again reduce to analyzing the behavior of the generating polynomial at points of the form provided in Equation~\eqref{eq:2-local-boundary-points}.

\begin{proof}[Proof of Lemma~\ref{lem:ly-to-sf}]
Take $(z_\theta, w_\theta)$ as in Equation~\eqref{eq:2-local-boundary-points}. In particular, let
\begin{align*}
    &z_\theta^{ij} = (0, \ldots, 0, \overbrace{e^{i \theta_1}}^{\text{slot } i}, 0, \ldots, 0, \overbrace{e^{i \theta_2}}^{\text{slot } j}, 0, \ldots, 0), \\
    &w_\theta^{ij} = (0, \ldots, 0, \underbrace{-e^{-i \theta_1}}_{\text{slot } i}, 0, \ldots, 0, \underbrace{-e^{-i \theta_2}}_{\text{slot } j}, 0, \ldots, 0),
\end{align*}
where $\theta_1, \theta_2 \in [0, 2\pi)$ are arbitrary and $\{i, j\} \in E$. These points are chosen precisely so that the generating polynomial of any term in $H$ that has an identity on site $i$ or site $j$ evaluates to zero. Our computation preceding this proof then shows that
\begin{align*}
    F_H(z_\theta^{ij}, w_\theta^{ij}) &= - \sum_{\alpha, \gamma \in \{x, y, z\}} J_{ij}^{\alpha \gamma} F_{\sigma_{\alpha}^i}(e^{i \theta_1}, -e^{-i \theta_1})F_{\sigma_{\gamma}^j}(e^{i \theta_2}, -e^{-i \theta_2}) \\
    &= -4 J_{ij}^{zz} + 4 \begin{pmatrix}
       \sin(\theta_1) \\ \cos(\theta_1) 
    \end{pmatrix}^\top \begin{pmatrix}
    J_{ij}^{xx} & J_{ij}^{xy} \\ J_{ij}^{yx} & J_{ij}^{yy}
    \end{pmatrix}\begin{pmatrix}
       \sin(\theta_2) \\ \cos(\theta_2) 
    \end{pmatrix} \\
    &\qquad - 4i \begin{pmatrix}
       \sin(\theta_1) \\ \cos(\theta_1) 
    \end{pmatrix}^\top \begin{pmatrix}
    J_{ij}^{xz} \\ J_{ij}^{yz}
    \end{pmatrix} - 4i \begin{pmatrix}
       \sin(\theta_2) \\ \cos(\theta_2) 
    \end{pmatrix}^\top \begin{pmatrix}
    J_{ij}^{zx} \\ J_{ij}^{zy}
    \end{pmatrix}.
\end{align*}
Now, because $\exp(-\beta H) \in \LY(1)$ for all $\beta \geq 0$, the contrapositive of Lemma~\ref{lem:2-local-perturbative} implies that $F_{H}(z_\theta^{ij}, w_\theta^{ij}) \in (-\infty, 0]$ for all points $(z_\theta, w_\theta)$. In turn, the imaginary part of the righthand side of the above expression must be zero for all $\theta_1, \theta_2 \in [0, 2\pi)$. Evaluating at $\theta_1 = \theta_2 = \pi / 2$ first and $\theta_1 = 3\pi / 2$, $\theta_2 = \pi / 2$ second shows that
\begin{align*}
    J_{ij}^{xz} + J_{ij}^{zx} = 0 = -J_{ij}^{xz} + J_{ij}^{zx}
\end{align*}
so that $J_{ij}^{xz} = 0$ upon rearranging. By identical reasoning, we also find that $J_{ij}^{yz} = J_{ij}^{zx} = J_{ij}^{zy} = 0$. The fact that $H$ is Hermitian requires that the remaining couplings constants are real.

The fact that $F_{H}(z_\theta^{ij}, w_\theta^{ij}) \in (-\infty, 0]$ then implies that
\begin{align*}
    J_{ij}^{zz} \geq \begin{pmatrix}
       \sin(\theta_1) \\ \cos(\theta_1) 
    \end{pmatrix}^\top \begin{pmatrix}
    J_{ij}^{xx} & J_{ij}^{xy} \\ J_{ij}^{yx} & J_{ij}^{yy}
    \end{pmatrix}\begin{pmatrix}
       \sin(\theta_2) \\ \cos(\theta_2) 
    \end{pmatrix} 
\end{align*}
for all $\theta_1, \theta_2 \in [0, 2\pi)$. More concisely,
\begin{align*}
    J_{ij}^{zz} \geq u^\top \begin{pmatrix}
    J_{ij}^{xx} & J_{ij}^{xy} \\ J_{ij}^{yx} & J_{ij}^{yy}
    \end{pmatrix}v
\end{align*}
for all unit vectors $u,v \in \R^2$. In turn,
\begin{align*}
    J_{ij}^{zz} \geq \sup_{\|u\| = \|v\| = 1} u^\top \begin{pmatrix}
    J_{ij}^{xx} & J_{ij}^{xy} \\ J_{ij}^{yx} & J_{ij}^{yy}
    \end{pmatrix}v = \sup_{\|v\| = 1} \left\|\begin{pmatrix}
    J_{ij}^{xx} & J_{ij}^{xy} \\ J_{ij}^{yx} & J_{ij}^{yy}
    \end{pmatrix}v\right\| = \left\|\begin{pmatrix}
    J_{ij}^{xx} & J_{ij}^{xy} \\ J_{ij}^{yx} & J_{ij}^{yy}
    \end{pmatrix}\right\|.
\end{align*}
The sites $i$ and $j$ were chosen arbitrarily among edges $\{i, j\} \in E$, so all that remains to check to conclude that $H$ is Suzuki-Fisher is whether the field strengths satisfy $\mu_i^z \geq 0$ for all $i \in V$. This step is the content Lemma~\ref{lem:ly-to-positive-field}, completing the proof.
\end{proof}

\section{Miscellaneous Computations}
\label{app:computations}
For the EPR-like Hamiltonian, it is useful to have a cleaner characterization of the Lee-Yang radius of the corresponding operator on a single edge, up to leading order.
\begin{lemma}
\label{lem:epr-single-edge-comp}
For $\beta \in (0, 1)$ and $s \in [0, 1)$, we can write
\begin{align*}
    \log r_\beta(s) = \left(\frac{1 - s^2}{4}\right) \beta + O(\beta^2).
\end{align*}
\end{lemma}
\begin{proof}
Borrowing the notation from \cite{wong2026}, we write
\begin{align*}
    \varepsilon_\beta(s) := \frac{e^{\beta(1 + s^2)} - 1}{1 + s^2}.
\end{align*}
Then, the radius arising in Lemma~\ref{lem:epr-like-edge-radius} becomes
\begin{align*}
    r_\beta(s) = \left(\frac{1 + \varepsilon_\beta(s)}{1 + s^2 \varepsilon_\beta(s)}\right)^{1/4}.
\end{align*}
We then obtain
\begin{align*}
    \log r_\beta(s) = \frac{1}{4}\left(\log(1 + \varepsilon_\beta(s)) - \log(1 + s^2 \varepsilon_\beta(s))\right)
\end{align*}
Note that
\begin{align*}
    \varepsilon_\beta(s) = \beta + O(\beta^2),
\end{align*}
uniformly for $s \in [0, 1]$. Because $\log(1 + x) = x + O(x^2)$, we obtain
\begin{align*}
    \log r_\beta(s) &= \frac{1}{4}\left(\log(1 + \beta + O(\beta^2)) - \log(1 + s^2 \beta  + O(\beta^2))\right) \\
    &= \left(\frac{1 - s^2}{4}\right) \beta + O(\beta^2)
\end{align*}
as claimed.
\end{proof}

Above, we also require a loose bound on the generating polynomial of an operator on the unit polydisk, which we prove here for the sake of completeness.

\begin{lemma}
\label{lem:gen-poly-bound}
If $A \in \sB((\C^2)^{\otimes n})$, then for all $(z, w) \in \overline{\D}_1^n \times \overline{\D}_1^n$ such that both $z$ and $w$ have at most $k$ nonzero coordinates (e.g., as in Equation~\eqref{eq:2-local-boundary-points}), we have that
\begin{align*}
    |F_A(z, w)| \leq 2^k \|A\|.
\end{align*}
\end{lemma}
\begin{proof}
As noted in \cite{wong2026}, we may define
\begin{align*}
    v(z) := \bigotimes_{i=1}^n (\ket{0}_i + z_i^* \ket{1}_i)
\end{align*}
so that
\begin{align*}
    F_A(z, w) = \langle v(z) | A | v(w^*) \rangle.
\end{align*}
Upon rewriting the generating polynomial in this form, we can immediately see that
\begin{align*}
    |F_A(z, w)| \leq \|v(z)\| \|v(w)\| \|A\| \leq \left(\prod_{i=1}^n (1 + |z_i|^2)\right)^{1/2} \left(\prod_{i=1}^n (1 + |w_i|^2)\right)^{1/2} \|A\| \leq 2^k \|A\|,
\end{align*}
using the fact that $(z, w) \in \overline{\D_1^n \times \D_1^n}$ in the last step.
\end{proof}

In our proof of Theorem~\ref{th:k-local-obstruction}, we also use the following bound on the distance, in operator norm, between two Gibbs states conjugated by distinct unitaries $U, U'$.

\begin{lemma}
\label{lem:gibbs-state-bound}
If $U, U'$ are unitary operators and $\beta > 0$, then
\begin{align*}
    \|\exp(-\beta H(U)) - \exp(-\beta H(U'))\| \leq 2 \beta \|H\| \exp(\beta \|H\|) \|U - U'\|.
\end{align*}
\end{lemma}
\begin{proof}
First, note that for any bounded operator $A$,
\begin{align*}
    \|U^* A U - (U')^* A U'\| \leq 2 \|A\| \|U - U'\|.
\end{align*}
To apply this identity, expand
\begin{align*}
    \exp(-\beta H(U)) - \exp(-\beta H(U')) = \sum_{m=1}^\infty \frac{(-\beta)^m}{m!} (U^* H^m U - (U')^* H^m U').
\end{align*}
Hence, the aforementioned identity gives
\begin{align*}
    \|\exp(-\beta H(U)) - \exp(-\beta H(U'))\| &\leq 2 \sum_{m=1}^\infty \frac{\beta^m}{m!} \|H\|^m \|U - U'\| \\
    &= 2 (\exp(\beta \|H\|) - 1) \|U - U'\| \\
    &\leq 2 \beta \|H\| \exp(\beta \|H\|) \|U - U'\|,
\end{align*}
noting that $e^x - 1 \leq x e^x$ for all $x \in \R$.
\end{proof}
\end{document}